\documentclass[a4paper,UKenglish,cleveref, autoref, thm-restate]{lipics-v2021}
\usepackage{macros}

\title{Checking Presence Reachability Properties on Parameterized Shared-Memory Systems}
\titlerunning{Checking PRPs on Parameterized Shared-Memory Systems} 
\author{Nicolas Waldburger}{Univ Rennes, Inria, CNRS, IRISA, France}{nicolas.waldburger@irisa.fr}{}{}
\Copyright{Nicolas Waldburger}
\authorrunning{N.~Waldburger}
\category{} 

\relatedversion{}

\nolinenumbers 

\begin{document}
\maketitle

\begin{abstract}
We consider the verification of distributed systems composed of an arbitrary number of asynchronous processes. Processes are identical finite-state machines that communicate by reading from and writing to a shared memory. Beyond the standard model with finitely many registers, we tackle round-based shared-memory systems with fresh registers at each round. In the latter model, both the number of processes and the number of registers are unbounded, making verification particularly challenging.
The properties studied are generic presence reachability objectives, which subsume classical questions such as safety or synchronization by
expressing the presence or absence of processes in some states. In the more general round-based setting, we establish that the parameterized verification of presence reachability properties is \PSPACE-complete.
Moreover, for the roundless model with finitely many registers, we prove that the complexity drops down to \NP-complete and we provide several natural restrictions that make the problem solvable in polynomial time. 
\end{abstract}

\section{Introduction}
\label{sec:intro}
\textbf{Parameterized verification.}
Distributed systems consist of multiple processes running in parallel. Verification of such systems is a major topic of modern verification, because of how common these systems are and how difficult their verification has proven to be. Indeed, when multiple processes run asynchronously, the number of relevant interleavings to consider quickly becomes large. 
An intuitive approach for their verification is to fix the number of processes involved and try to apply classical verification techniques. Another approach is that of parameterized verification, where 
one aims to prove the more general statement that the property of interest holds for any number of participants. The interest of this approach is threefold. First, it allows to prove that the system is correct regardless of the number of processes. 
Second, the efficiency of parameterized techniques does not depend on the number of participants, which makes them more suitable for large systems for which classical techniques scale poorly. Third, parameterized verification often yields decidability or better computational complexity for problems that are hard to solve with classical techniques; see for example \cite{hague} for a problem that becomes decidable in the parameterized case.
In their seminal work \cite{GS-jacm92}, German and Sistla consider systems consisting of a leader and arbitrarily many contributors, all of which are finite-state machines communicating via \emph{rendez-vous}. In this setting, the safety verification problem is \EXPSPACE-complete and the complexity drops down to polynomial time when the leader is removed. Since then, many similar models have been studied, with variations on the expressiveness of the processes and the means of communication in order to capture the large variety of existing distributed algorithms \cite{Esparza-stacs14,BJKKRVW-book15}.

\textbf{Contributions.}
We study parameterized verification of systems where all processes are identical and anonymous finite-state machines that communicate via reading from and writing to a shared memory. The read and write actions are performed non-atomically, meaning that no process may perform a read-write combination while preventing all other processes from acting. Our registers are \emph{initialized} with a special symbol; this assumption is common in parameterized verification of shared-memory systems \cite{BMRSS-icalp16,tso}, since some algorithms require initialized registers, \emph{e.g.} \cite{Aspnes-ja02}. 
First, we study a model with finitely many registers. This model is inspired by \cite{EGM-jacm16} where registers were uninitialized and the verification is restricted to safety properties. In contrast, we study the more general \emph{presence reachability problems}, in which one asks whether one may reach a configuration that satisfies a property. This property takes the form of a Boolean combination of constraints expressing whether there is at least one process in a given state of the finite-state machine. We prove that this problem is \NP-complete and we provide several natural restrictions on the process description and on the property that make the problem solvable in polynomial time. 
We then work on the more general setting of \emph{round-based} shared-memory systems \cite{myicalppaper}, which are designed to model round-based shared-memory algorithms present in the literature, see \emph{e.g.} \cite{Aspnes-ja02,raynalSimpleAsynchronousShared2012}. In this model, the processes proceed in asynchronous rounds, each round having its own fresh set of registers. The source of infinity is twofold, as the number of processes and the number of registers are both unbounded, making round-based systems particularly challenging to verify. The safety problem was proved to be \PSPACE-complete in round-based shared-memory systems \cite{myicalppaper}. In this article, we go beyond safety by considering a round-based, richer version of the presence reachability problem where the property may quantify existentially and universally over the rounds. 
Nonetheless, we establish that the round-based presence reachability problem is \PSPACE-complete.

\textbf{Related work.}
Similar models and problems have been studied in the literature. In the shared-memory model (without rounds and without register initialization), the safety problem has been studied extensively with variations on the expressiveness given to the leader and the contributors \cite{EGM-jacm16}; in particular, when processes are finite-state machines, the safety problem is shown to be \coNP-complete and to decrease to \PTIME when the leader is removed. However, this result does not hold when registers are initialized or when the property is more general than safety. 
A model that has perhaps been more studied is that of \emph{reconfigurable broadcast networks} (RBN), where processes communicate via broadcasting messages that can be received by any of the other processes. This model has similarities with shared-memory systems 
, although broadcast tends to be simpler (messages disappear after being sent, while written values remain in the registers).
A source of inspiration for the first part of our article is the study of reachability problems in RBN \cite{delzanno2012}, where it is shown that the cardinality reachability problem, where one wants to reach a configuration that satisfies cardinality constraints, is \PSPACE-complete. When the constraints cannot count processes, this problem is analogous to our presence reachability problem; for RBN, it is shown to be \NP-complete, a complexity that we also obtain in our setting. Finally, this complexity drops down to \PTIME{} in RBN when considering the special case of safety. This tractability result no longer holds in the shared-memory world unless we make further assumptions about the number of registers or their initialization.
The cube reachability problem is a generalization of the cardinality constraint problem  where the initial configuration is also subject to cardinality constraints; this problem is \PSPACE-complete both in RBN and in (roundless) asynchronous shared-memory systems \cite{delzanno2012, BW-gandalf21,  BW-fossacs21}, although it is unknown whether this remains true when allowing the $\textsf{Pre}^*$ and $\textsf{Post}^*$ operators in the description of the cubes \cite{BW-fossacs21, BW-erratum}. While it is interesting to compare results on RBN with our results on shared-memory systems without rounds, such a comparison is not possible with the more expressive model of round-based shared-memory systems, 
in particular because the unboundedness in the number of registers has no equivalent in broadcast networks. 

Due to space constraints, most of the proofs can be found in the appendix.
\begin{scope}
\section{Roundless Register Protocols}
\label{sec:roundless_prots}
In this section, we introduce \emph{"register protocols@@rl"}, a model inspired by \cite{EGM-jacm16}. We call these systems \emph{roundless} to distinguish them from \emph{round-based} systems introduced later in this article. 
\subsection{Definitions}
\label{subsec:defs_roundless}
\begin{definition}[Roundless register protocols]
  A ""roundless register protocol@@rl"" is a tuple $\prot =
    \nobreak{\tuple{\states, \initialstates,\rdim, \dataalp, \datainit, \transitions}}$
  where
  \begin{itemize}
    \item $\states$ is a finite set of ""states@@rl"" with a distinguished subset of 
          ""initial states@@rl""~$\initialstates \subseteq \states$;
    \item $\rdim \in \nats$ is the number of shared ""registers@@rl""; 
    \item $\dataalp$ is a finite ""data
          alphabet@@rl"" containing the initial symbol $\datainit$;
    \item
          $\transitions \subseteq \states \times \actions \times \states$ is
          the set of ""transitions@@rl"", where
          $\actions :=  \{\rlreadact{\regid}{\asymb} \mid \regid \in \regint,\allowbreak \asymb \in \dataalp\} \cup \{
            \writeact{\regid}{\asymb} \mid \regid \in \regint, \asymb \in \datawrite\}$
          is the set of ""actions@@rl"".
  \end{itemize}
\end{definition}
"Roundless register protocols@@rl" are executed on multiple processes that behave asynchronously and can only communicate via reading from and writing to the shared registers. The behavior of a process is described by a finite-state machine. The possible actions of the transitions are reading a symbol from and writing a symbol to one of the $\rdim$ shared registers; $\asymb \in \dataalp$ denotes the symbol and $\regid$ indicates the register on which the action is performed. Each register stores one ""symbol"" from the finite set $\dataalp$ at a time. Read-write combinations are performed non-atomically, i.e., no process can perform a read-write combination while excluding all other processes. The ""size@@prot"" of the protocol $\prot$ is defined as $|\prot| := |\states| + |\dataalp| + |\transitions| + \rdim$. For all $\regid \in \regint$, we write $\rlreg{\regid}$ for the "register" of index $\regid$. We also write $\regset{}$ for the set $\set{\rlreg{\regid} \mid \regid \in \regint}$ of all "registers". 

Processes are assumed to have no identifiers so they are identical anonymous agents. Therefore, a ""configuration@@arl"" is a pair $\cconfig = \confpair{\mu}{\dvec} \in \NN^{\states} \times \dataalp^{\regset{}}$ such that $0 < \sum_{q \in \states} \mu(q) < \infty$.  Let $\state{\cconfig} := \mu$ which indicates the number of processes in each state, and $\datafun{\cconfig} := \dvec$ mapping to each register its symbol: for all $\regvar \in \regset{}$, $\data{\cconfig}{\regvar}$ is the symbol contained in "register@@rl" $\regvar$ in $\cconfig$. Let $\cconfigs := \NN^{\states} \times \dataalp^{\regset{}}$ denote the set of all "configurations@@arl". Let $\supp{\cconfig} := \set{q \in Q \mid \state{\cconfig}(q) >0}$ denote the support of the multiset $\state{\cconfig}$. We write $\intro{\oplus}$ and $\intro{\ominus}$ the operations on multisets that add and remove elements, respectively.
A configuration is ""initial@@arl"" if all processes are in states from $\initialstates$ while all "registers@@rl" have value $\datainit$.
We denote by $\cconfiginitset$ the set of "initial configurations@@arl" (the letter $\mathsf{c}$ stands for ``concrete'' as opposed to ``abstract'' configurations defined later). Formally,
$\cconfiginitset := \set{\cconfig \mid \state{\cconfig} \subseteq \initialstates, \, \datafun{\cconfig} = \datainit^{\regset{}}}$.

Given $\cconfig, \cconfig' \in \cconfigs$, $\cconfig'$ is a ""successor@@concrete"" of $\cconfig$ when there exists $\atrans = (q,a,q') \in \transitions$ such that $\state{\cconfig}(q) >0$, $\state{\cconfig'} = (\state{\cconfig} \ominus \set{q}) \oplus \set{q'}$ and:
\begin{itemize}
\item if $a = \rlreadact{\regid}{\asymb}$ then $\data{\cconfig}{\rlreg{\regid}} = \asymb$ and $\datafun{\cconfig'} = \datafun{\cconfig}$,
\item if $a = \writeact{\regid}{\asymb}$ then $\data{\cconfig'}{\rlreg{\regid}} = \asymb$ and $\forall \regid' \ne \regid$, $\data{\cconfig'}{\rlreg{\regid'}} = \data{\cconfig}{\rlreg{\regid'}}$. 
\end{itemize}
In that case, we write $\cconfig \step{\atrans} \cconfig'$ or simply $\cconfig \step{} \cconfig'$, which is called a ""step"". . 
A ""concrete execution"" is a sequence $\cexec = \cconfig_0, \atrans_1, \cconfig_1, \dots, \cconfig_{l-1}, \atrans_{l}, \cconfig_l$ such that for all $i$, $\cconfig_{i} \step{\atrans_{i+1}} \cconfig_{i+1}$. We write $\cconfig_0 \step{*} \cconfig_l$ for the existence of such an execution. 
$\cconfig'$ is ""reachable from $\cconfig$"" when $\cconfig \step{*} \cconfig'$. Given a set $C$ of configurations, we write $\creach{C} := \set{\cconfig' \mid \exists \cconfig \in C, \, \cconfig \step{*} \cconfig'}$.
A configuration is ""reachable"" when it is in $\creach{\cconfiginitset{}}$.   

\begin{example}
\label{example:roundless_protocol}
Figure~\ref{fig:example_roundless_protocol} provides an example of a "roundless register protocol" $\prot$ with $\dataalp = \set{\datainit, \exa, \exb, \exc}$, $\initialstates = \set{q_0}$ and $\rdim = 1$, hence read and write actions are implicitly on register $\regid = 1$. The red and blue labels are to be ignored for now. 

The set of initial configurations is $\cconfiginitset := \set{\confpair{q_0^n}{\datainit} \mid n \geq 1}$. 
The following execution with two processes witnesses that $\confpair{\errorstate \oplus C}{\exa} \in \creach{\cconfiginitset{}}$: \\
 $\confpair{q_0^2}{\datainit} \step{(q_0, \rlreadact{}{\datainit}, B)} \confpair{q_0 \oplus B}{\datainit} \step{(B, \rlreadact{}{\datainit}, C)} \confpair{q_0 \oplus C}{\datainit} \step{(q_0, \writeact{}{\exc}, A)} \smallskip \\ \confpair{A \oplus C}{\exc} \step{(C, \writeact{}{\exa}, C)} \confpair{A \oplus C}{\exa} \step{(A, \rlreadact{}{\exa}, \errorstate)} \confpair{\errorstate \oplus C}{\exa}$.
\begin{figure}[tp]
\centering
\begin{tikzpicture}[x = 1.5cm, y= 0.8cm, node distance = 0.2cm]
\tikzstyle{every node} = [font = \small]
\node[state] (q0) at (0,0) {$q_0$};
\node[state] (B) at (0,-2) {$B$};
\node[state] (A) at (2,0) {$A$};
\node[state] (C) at (2,-2) {$C$};
\node[state] (qf) at (4,-1) {$\errorstate$};

\path[-latex', thick]
(q0.180) edge[latex'-] +(180:4mm)
(q0) edge node[left, name = readb, overlay, yshift = 0.15cm] {$\rlreadact{}{\datainit$}} (B)
(q0) edge node[above] {$\writeact{}{\exc}$} (A)
(B) edge node[below] {$\rlreadact{}{\datainit}$} (C)
(C) edge node[left, name = readc, overlay, yshift = 0.15cm] {$\rlreadact{}{\exc}$} (A)
(A) edge[bend left = 15] node[above, sloped] {$\rlreadact{}{\exa}$} (qf)
(qf) edge[bend left = 15] node[below, sloped, yshift = 0.08cm] {$\writeact{}{\exb}$} (A)
(C) edge node[below, sloped, bend right = 15] {$\rlreadact{}{\exb}$} (qf)
(C) edge[loop below] node[right, yshift = 0.2cm, xshift = 0.1cm] {$\writeact{}{\exa}$} ();   

\node[below = of readb, text = blue!50!white, yshift = 0.35cm] {$\rlreadact{}{\exc}$};
\node[below = of readc, text = red!50!white, yshift = 0.35cm] {$\writeact{}{\exa}$};
\end{tikzpicture}
\caption{An example of a protocol}
\label{fig:example_roundless_protocol}
\end{figure}
\end{example}

\subsection{Reachability Problems}
\label{subsec:rl_problems}

Our first problem of interest is the ""coverability problem"" (\COVER): 
\smallskip

\noindent\fbox{\begin{minipage}{.99\linewidth}
    \textsc{\COVER for roundless register protocols} \\
    {\bf Input}: A "roundless register protocol" 
    $\prot$, $\errorstate \in \states$\\
    {\bf Question}: Does there exist $\cconfig \in \creach{\cconfiginitset}$ such that $\state{\cconfig}(\errorstate)>0$?
  \end{minipage}} \smallskip

Note that, because the model is parameterized, a witness execution of \COVER may have an arbitrarily large number of processes. 
The dual is the ""safety problem"", the answer to which is yes when an error state cannot be covered regardless of the number of processes.
A similar problem is the ""target problem"" (\TARGET) where processes must synchronize at $\targetstate$ : \smallskip

\noindent\fbox{\begin{minipage}{.99\linewidth}
    \textsc{\TARGET for roundless register protocols} \\
    {\bf Input}: A "roundless register protocol" 
    $\prot$, $\targetstate \in \states$\\
    {\bf Question}: Does there exist $\cconfig \in \creach{\cconfiginitset}$ s.t. for all $q \ne \targetstate$, $\state{\cconfig}(q)=0$?
  \end{minipage}} \smallskip

\begin{remark}
\label{rem:target_esier_than_cover} \TARGET is harder than \COVER:
consider the reduction in which one adds a loop on $\errorstate$ writing a joker symbol which, from any state, may be read to reach $\errorstate$. 
\end{remark}

""Presence constraints"" are Boolean combinations (with $\land$, $\lor$ and $\neg$) of "atomic propositions" of the form ``$q \populated$'' with $q \in \states$, or of the form ``$\regvar \contains \asymb$'' with $\regvar \in \regset{}$ and $\asymb \in \dataalp$. A "presence constraint" is interpreted over a "configuration@@concrete" $\cconfig$ by interpreting ``$q \populated$'' as true if and only if $\state{\cconfig}(q) > 0$ and ``$\regvar \contains \asymb$'' as true if and only if $\data{\cconfig}{\regvar} = \asymb$. Note that "presence constraints" cannot refer to how many processes are on a given state.
We write $\cconfig \models \phi$ when "configuration@@concrete" $\cconfig$ satisfies "presence constraint" $\phi$. 

\begin{example}
If $\states = \set{q_1, q_2, q_3}$, $\rdim = 2$, $\dataalp = \{\datainit, a,b\}$ and $\phi :=  (q_1 \populated) \lor ((q_2 \populated) \land (\rlreg{1} \contains a))$ then $\confpair{q_1 \oplus q_3}{\datainit^{2}} \models \phi$, $\confpair{q_2^2}{(a,b)} \models \phi$ but $\confpair{q_2^2}{b^2} \not \models \phi$.  
\end{example}

The ""Presence Reachability Problem"" (\PRP) generalizes both \COVER and \TARGET. It corresponds to the cardinality reachability problem for cardinality constraints restricted to CC[$\geq 1, =0$] studied for broadcast protocols \cite{delzanno2012}. \smallskip

\noindent\fbox{\begin{minipage}{.99\linewidth}
    \textsc{\PRP for roundless register protocols} \\
    {\bf Input}: A roundless register protocol 
    $\prot$, a "presence constraint" $\phi$\\
    {\bf Question}: Does there exist $\cconfig \in \creach{\cconfiginitset}$ such that $\cconfig \models \phi$?
  \end{minipage}} \smallskip

The formula $\phi$ automatically makes \PRP \NP-hard, since one can encode the \SAT problem. Therefore, we also consider the ""DNF Presence Reachability Problem"" (\dnfPRP), in which $\phi$ is in ""disjunctive normal form"".
\COVER and \TARGET are special cases of \dnfPRP, with $\phi = (\errorstate \populated)$ for \COVER and $\phi = \bigland_{q \ne \targetstate} \neg (q \populated)$ for \TARGET. 

\begin{example}
\label{example:roundless_reachabilities}
Consider again the protocol $\prot$ defined in Figure~\ref{fig:example_roundless_protocol}. $(\prot, q_f)$ is a positive instance of \COVER, as proved in Example~\ref{example:roundless_protocol}. Let $\prot_{\mathsf{blue}}$ be the protocol obtained from $\prot$ by changing to $\rlreadact{}{\exc}$ the label of the transition from $q_0$ to $B$ (in blue in Figure~\ref{fig:example_roundless_protocol}). $(\prot_{\mathsf{blue}}, q_f)$ is a negative instance of \COVER. In fact, a process can only get to $B$ if $\exc$ has been written to the register, and then $\datainit$ can no longer be read so no process may go to state $C$, $\exa$ cannot be written and no process may go from $A$ to $\errorstate$. 

$(\prot, q_f)$ is a negative instance of \TARGET: to leave $A$, one needs to read $\exa$, hence must have a process on state $C$, and to leave $C$, one must read $\exb$ which would force us to send a process to $A$. Let $\prot_{\mathsf{red}}$ be the protocol obtained from $\prot$ by changing to $\writeact{}{\exa}$ the label of the transition from $C$ to $A$ (in red in Figure~\ref{fig:example_roundless_protocol}). $(\prot_{\mathsf{red}}, q_f)$ is a positive instance of \TARGET:
$\confpair{q_0^2}{\datainit} \step{(q_0, \rlreadact{}{\datainit}, B)} \confpair{q_0 \oplus B}{\datainit} \step{(B, \rlreadact{}{\datainit}, C)} \confpair{q_0 \oplus C}{\datainit} \step{(q_0, \writeact{}{\exc}, A)} \\\confpair{A \oplus C}{\exc} \step{(C, \writeact{}{\exa}, A)} \confpair{A^2}{\exa} \step{(A, \rlreadact{}{\exa}, \errorstate)} \confpair{A \oplus \targetstate}{\exa} \step{(A, \rlreadact{}{\exa}, \errorstate)} \confpair{\targetstate^2}{\exa}$.
\smallskip


\noindent Let $\phi := \neg (C \populated) \land ((\onlyreg \contains \exa) \lor [(\onlyreg \contains \exb) \land \neg(A \populated)])$. $\phi$~is a "presence constraint" and $(\prot, \phi)$ is a negative instance of \PRP. Indeed, if $\exa$ is in the register, then $C$ must be populated and if $\exb$ is in the register, then $A$ must be populated.
\end{example}

\subsection{Abstract Semantics}
\label{subsec:abstract_semantics_roundless}

In this subsection, we define an abstraction of the semantics that is sound and complete with respect to \PRP. The intuition of this abstraction is that the exact number of processes in a given state is not relevant. Indeed, "register protocols", thanks to non-atomicity, enjoy a classical monotonicity property named copycat property. 

\begin{restatable}[Copycat]{lemma}{copycat}
\label{lem:copycat_roundless}
Consider $\cconfig_1$, $\cconfig_2$, $q_2$ such that $\cconfig_1 \step{*} \cconfig_2$, $q_2 \in \supp{\cconfig_2}$.
There exists $q_1 \in \supp{\cconfig_1}$ s.t. $\confpair{\state{\cconfig_1} \oplus q_1}{\datafun{\cconfig_1}} \step{*} \confpair{\state{\cconfig_2} \oplus q_2}{\datafun{\cconfig_2}}$. 
\end{restatable}

An ""abstract configuration"" is a pair $\aconfig = \confpair{\state{\aconfig}}{\datafun{\aconfig}} \in 2^{\states} \times \dataalp^{\regset{}}$ such that $\state{\aconfig} \ne \emptyset$. 
The set of ""initial configurations@@abstract"" is $\aconfiginitset := \set{\confpair{S}{\datainit^{\rdim}} \mid S \subseteq \initialstates}$. 
Given a "concrete configuration" $\cconfig$, the "projection" $\abstractproj{\cconfig}$ is the "abstract configuration" $\confpair{\supp{\cconfig}}{\datafun{\cconfig}}$. Let $\aconfigs:=2^{\states} \times \dataalp^{\regset{}}$ denote the set of abstract configurations.
For $\aconfig, \aconfig' \in \aconfigs$, $\aconfig'$ is the ""successor@@abstract"" of $\aconfig$ when there exists $\atrans = (q, a,q') \in \transitions$ such that $q \in \state{\cconfig}$,
either $\state{\cconfig'} = \state{\cconfig} \cup \set{q'}$ 
or $\state{\cconfig'} = (\state{\cconfig} \setminus \set{q}) \cup \set{q'}$,
and: if $a = \rlreadact{\regid}{\asymb}$ then $\data{\cconfig}{\rlreg{\regid}} = \asymb$ and $\datafun{\aconfig} = \datafun{\aconfig'}$, and if $a = \writeact{\regid}{\asymb}$ then  $\data{\aconfig'}{\rlreg{\regid}} = \asymb$ and for all $\regid' \ne \regid$, $\data{\aconfig'}{\rlreg{\regid'}} = \data{\aconfig}{\rlreg{\regid'}}$. 
Again, we denote such a step by $\aconfig \step{\atrans} \aconfig'$ or $\aconfig \step{} \aconfig'$. Note that one could equivalently define $\aconfig \step{\atrans} \aconfig'$ by: $
\aconfig \step{\atrans} \aconfig' \Longleftrightarrow \exists \cconfig, \cconfig' \in \cconfigs, \, \cconfig \step{\atrans} \cconfig' \text{ and } \abstractproj{\cconfig} = \aconfig, \abstractproj{\cconfig'} = \aconfig'$. 
This notion of abstraction is classical in parameterized verification of systems with identical anonymous agents that enjoy monotonicity properties. 
Note, however, that this semantics is non-deterministic: one could have $\aconfig'' \ne \aconfig'$ such that $\aconfig \step{\atrans} \aconfig'$ and $\aconfig \step{\atrans} \aconfig''$. This alternative corresponds to whether all processes in $q$ take the transition ($\state{\cconfig'} = (\state{\cconfig} \setminus \set{q}) \cup \set{q'}$) or only some ($\state{\cconfig'} = \state{\cconfig} \cup \set{q'}$).
We define ""abstract executions"" similarly to concrete ones, and denote them using $\exec$. We also define the ""reachability set"" $\areach{A}$ and the notion of ""coverability"" as in the concrete case.  
This  abstraction is sound and complete for \PRP:

\begin{restatable}[Soundness and completeness of the abstraction]{proposition}{soundcompleteroundless}
  \label{prop:soundcomplete_roundless}
For all 
$S \subseteq \states$, $\dvec \in \dataalp^{\regset{}}$: 
\[
    (\exists \cconfig \in \creach{\cconfiginitset{}}:
    \supp{\cconfig} {=} S, \, \datafun{\cconfig} {=} \dvec)
     \  \Longleftrightarrow \\   (\exists \aconfig \in \areach{\aconfiginitset{}}:\
    \state{\aconfig} {=} S, \, \datafun{\aconfig} {=} \dvec). 
  \]
\end{restatable}
The intuition of the proof is the following: any concrete configuration can easily be lifted into an abstract one. Conversely, any abstract execution may be simulated in the concrete semantics for a sufficiently large number of processes by using the copycat property.

Given a "presence constraint" $\phi$ and $\aconfig \in \aconfigs$, we define whether $\aconfig$ satisfies $\phi$, written $\aconfig \models \phi$, in a natural way. 
Given a concrete configuration $\cconfig$, one has $\cconfig \models \phi$ if and only if $\abstractproj{\cconfig} \models \phi$. Indeed, $\cconfig$ and $\abstractproj{\cconfig}$ have the same "populated states" and "register values".  Therefore, there exists $\cconfig \in \creach{\cconfiginitset}$ such that $\cconfig \models \phi$ if and only if there exists $\aconfig \in \areach{\aconfiginitset}$ such that $\aconfig \models \phi$: one can consider \PRP directly in the abstract semantics.




\section{Complexity Results for Roundless Register Protocols}
\label{sec:roundless_results}
In this section, we provide complexity results for the presence reachability problems defined above in the general case and in some restricted cases. Throughout the rest of the section, all configurations and executions are implicitly "abstract". 

\subsection{\NP-Completeness of the General Case}
\label{subsec:roundless_np_completeness}
First, all problems defined in the previous section are \NP-complete. 

\begin{restatable}{proposition}{prpnpc}
\label{prop:prp_npc_roundless}
\COVER, \TARGET, \dnfPRP and \PRP for roundless register protocols are all \NP-complete.
\end{restatable}
\begin{proof}
First, we prove that all four problems are in \NP. It suffices to prove it for \PRP, as the three other problems reduce to it.

Let $\exec: \aconfig_0 \step{*} \aconfig$ an abstract execution, we simply prove the existence of $\exec': \aconfig_0 \step{*} \aconfig$ of length at most $4 |Q|$. To obtain $\exec'$ from $\exec$, we iteratively:
\begin{itemize}
\item remove any read step that is "non-deserting" and does not cover a new location,
\item remove any write step that is "non-deserting", does not "populate" a new state and whose written symbol is never read,
\item make "non-deserting" any "deserting" step whose "source" state is  "populated" again later in $\exec$.
\end{itemize}
In $\exec'$, at most $|Q|$ steps populate a new state and at most $|Q|$ steps are "deserting". This implies that there are at most $2|Q|$ read steps, therefore, at most $2|Q|$ write steps whose written value is actually read. In total, this bounds the number of steps by $4 |Q|$. In particular, for \PRP, we can look for an execution of length less than $4 |Q|$ which can be guessed in polynomial time.

\begin{figure}[tb]
   \centering
\resizebox{0.95\linewidth}{!}{
    \begin{tikzpicture}[node distance = 1.5cm, auto, y = 0.6cm, x = 0.7cm]
\tikzstyle{every node}=[font=\small]
\tikzstyle{every state} = [minimum size = 0.8cm, inner sep = 0.05cm]

\node (q0) [state] at (-0.25,0) {$q_0$}; 
\node(testc1) [state] at (1.5,0) {$C_1?$};
\node (c1ok) [state] at (6,0) {$C_2?$};
\node (susp) at (8.125,0) {{\large \dots}};
\node (testcm) [state] at (10.25,0) {$C_m?$};
\node (qf) [state] at (14.75,0) {$\errorstate$};
\node (int1) [state, rectangle] at (3.75,2) {$\testsat{l_{1,1}}$}; 
\node (int2) [state, rectangle] at (3.75,0) {$\testsat{l_{1,2}}$}; 
\node (int3) [state, rectangle] at (3.75,-2) {$\testsat{l_{1,3}}$}; 
\node (int1bis) [state, rectangle] at (12.5,2) {$\testsat{l_{m,1}}$}; 
\node (int2bis) [state, rectangle] at (12.5,0) {$\testsat{l_{m,2}}$}; 
\node (int3bis) [state, rectangle] at (12.5,-2) {$\testsat{l_{m,3}}$}; 
\node (subprot) [state, rectangle, minimum size = 0.8cm, inner sep = 0.1cm] at (1, -4) {$\testsat{l}$};
\node (assign) at (2.3,-4) {$:=$};
\node (test1) [state] at (4,-4) {};
\node (test2) [state] at (8,-4) {};
\node (test3) [state] at (12,-4) {};

\path[-latex', thick]
    (q0.180) edge[latex'-] +(180:4mm)
    (q0) edge[loop above] node[align = center]{{\small $\forall j \in \nset{n}$} \\ {\small $\writeact{\regsat{x_j}}{\sattrue}$}} ()
    (q0) edge[loop below] node[align = center]{{\small $\writeact{\regsat{\neg x_j}}{\sattrue}$}\\{\small $\forall j \in \nset{n}$}} ()
    (q0) edge (testc1)
    (testc1) edge (int1)
    (testc1) edge (int2)
    (testc1) edge (int3)
    (testcm) edge (int1bis)
    (testcm) edge (int2bis)
    (testcm) edge (int3bis)
    (int1bis) edge (qf)
    (int2bis) edge (qf)
    (int3bis) edge (qf)
    (int1) edge (c1ok)
    (int2) edge (c1ok)
    (int3) edge (c1ok)
    (test1.180) edge[latex'-] +(180:4mm)
    (test3.0) edge[-latex'] +(0:4mm)   
    (test1) edge node[above] {{\small $\rlreadact{\regsat{l}}{\sattrue}$}} (test2)
    (test2) edge node[above] {{\small $\rlreadact{\regsat{\neg l}}{\datainit}$}} (test3);

\path[dashed] 
    (c1ok) edge (6 + 0.5*3, 0.5*3)
    (c1ok) edge (6+0.5*3, 0)
    (c1ok) edge (6 + 0.5*3, -0.5*3);
\path[dashed, -latex']
    (10.25-0.5*3, 0.5*3) edge (testcm)
    (10.25-0.5*3, 0) edge (testcm)
    (10.25-0.5*3, -0.5*3) edge (testcm);
\end{tikzpicture}
}    
\caption{The protocol $\protsat{\phi}$ for \NP-hardness of \COVER.}
\label{fig:prot_3SAT_to_cover}
\end{figure}
We now prove \NP-hardness of \COVER, as it reduces to the three other problems.


The proof is by a reduction from 3-\SAT. Consider a 3-CNF formula $\phi = \bigland_{i=1}^{m} l_{i,1} \lor l_{i,2} \lor l_{i,3}$ over $n$ variables $x_1, \dots, x_n$
where, for all $i \in \nset{m}$, for all $k \in \nset{3}$, $l_{i,k} \in \set{x_j, \neg x_j \mid j \in \nset{n}}$.  
We define a "roundless register protocol" $\protsat{\phi}$ with a distinguished state $\errorstate$ which is "coverable" if and only if $\phi$ is satisfiable. In $\protsat{\phi}$, one has $\dataalp=\set{\datainit,\sattrue}$ and $\rdim = 2n$, there are two registers for each variable $x_i$, $\regsat{x_i}$ and $\regsat{\neg x_i}$. The protocol is represented on Figure~\ref{fig:prot_3SAT_to_cover}. 

While any register may be set to $\sattrue$ thanks to the loops on $q_0$, a register set to $\sattrue$ can never be set back to $\datainit$. $l$ is considered true if $\regsat{l}$ is set to $\sattrue$ while $\regsat{\neg l}$ still has value $\datainit$.

Suppose that the instance of 3-\SAT is positive, \emph{i.e.}, $\phi$ is satisfiable by some assignment $\nu$. Consider an execution that writes $\sattrue$ exactly to all $\regsat{l}$ with $l$ true in $\nu$. For each clause, one of the three literals is true in $\nu$. Therefore the execution may cover $C_i?$ for all $i$ so it may cover $\errorstate$ and the instance of \COVER is positive.
Conversely, if the instance \COVER is positive, there exists an execution $\exec : \aconfig_0 \step{*} \aconfig$ with $\aconfig_0 \in \aconfiginitset$ and $\errorstate \in \state{\aconfig}$. Consider $\nu$ that assigns to each variable $x$ value true if $\regsat{x}$ is written before $\regsat{\neg x}$ in $\exec$ and false otherwise. Given a litteral $l$, $\exec$ may only go through  $\testsat{l}$ if $\nu(l)$ is true; because $\exec$ covers $q_f$, this proves that $\nu \models \phi$.   
\end{proof}

\begin{remark}
In \cite{EGM-jacm16}, the authors prove \NP-completeness of \COVER in a similar model, but with a leader: in the \NP-hardness reduction, the leader make non-determinstic decisions about the values of the variable. This argument does not hold in the leaderless case. \end{remark}

\subsection{Interesting Restrictions}
\label{subsec:restrictions}
Although all the problems defined above are \NP-complete, they are sometimes tractable under appropriate restrictions on the protocols.
We will consider two restrictions on the protocols.  The first one is having $\rdim =1$, \emph{i.e.}, a single register.
The second restriction is the ""uninitialized case"" where processes are not allowed to read the initial value $\datainit$ from the registers. Formally, a protocol $\prot$ is ""uninitialized"" if its set of transitions $\transitions$ does not contain an action reading symbol $\datainit$: in "uninitialized" protocols, it is structurally impossible to read from an unwritten register. One might object that forbidding transitions that read $\datainit$ contradicts the intuition that, when a process reads from a register, it does not know whether the value is initial or not; one could settle the issue by considering that reading $\datainit$ sends processes to a sink state. 
The uninitialized setting tends to yield better complexity than the general, initialized case, see for example \cite[Section 7]{tso}.  

Of course, for \PRP, the formula itself always makes the problem \NP-hard.
\begin{proposition}
\label{prop:prp_np_hard_restrictions}
"\PRP" for "roundless register protocols" is \NP-hard even with $\rdim = 1$ and the register "uninitialized".
\end{proposition}

\subsection{Tractability of \COVER and \dnfPRP under Restrictions}
\label{subsec:tractability_cover_dnfprp}

In this subsection, we prove that \COVER is solvable in \PTIME when the "protocol@@roundless" is "uninitialized" or when $\rdim$ is fixed and that \dnfPRP is solvable is \PTIME when $\rdim = 1$. 


In \cite[Theorem 9.2]{EGM-jacm16}, "uninitialized" \COVER is proved to be \PTIME-complete; their approach, based on languages, is quite different from the one presented here. Our approach, similar to the one presented in \cite[Algorithm 1]{delzanno2012} in the setting of reconfigurable broadcast networks, is to compute the set of coverable states using a simple ""saturation"" technique, a fixed-point computation over the set of states. 

When registers are initialized, the saturation technique breaks down as it may be that some states are coverable but not in the same execution, as they require registers to lose their initial value in different orders (see the notion of first-write order developed in \cite{myicalppaper} for more development on this in a round-based setting). 
However, in the initialized case with a fixed number of registers, one can iterate over every such order and \COVER is tractable as well.

\begin{restatable}{proposition}{coverrestrictions}
\label{prop:ptime_cover_restrictions}
\COVER for "roundless register protocols" is \PTIME{}-complete either when the registers are "uninitialized" or when $\rdim$ is fixed.
\end{restatable}



For \dnfPRP, we provide a \PTIME algorithm in the more restrictive case of $\rdim = 1$.

\begin{restatable}{proposition}{ptimednfprponereg}
\label{prop:ptime_dnfprp_onereg_roundless}
\dnfPRP for "roundless register protocols" with $\rdim = 1$ is in \PTIME.
\end{restatable}
\begin{proof}[Proof sketch]
We give here the proof for \TARGET. See Appendix~\ref{appendix:ptime_dnfprp_onereg_roundless} for the proof and pseudocode for \dnfPRP.  
Our algorithm shares similarities with \cite[page 41]{fournier_phd_thesis} for broadcast protocols, although it is more complex because of the persistence of symbols in the register. 

First, we have a polynomial reduction from "initialized" \TARGET with $\rdim =1$ to "uninitialized" \TARGET with $\rdim = 1$. It proceeds as follows. Consider the graph $G= (Q,E)$ when $(q_1,q_2) \in E$ when there exists $(q_1, \rlreadact{}{\datainit}, q_2) \in \transitions$. Let $I \subseteq Q$ the set of states that are reachable in $G$ from $\initialstates$. The reduction simply replaces $\initialstates$ by $I$ as set of initial states.

Any (abstract) execution
$\exec: \aconfig_0 \step{*} \confpair{\targetstate}{\finalsymb}$, called ""synchronizing execution"", can be rearranged into 
$\exec_{+}: \aconfig_0 \step{*} \confpair{S}{\asymb}$ and 
$\exec_-: \confpair{S}{\asymb} \step{*} \confpair{\targetstate}{\finalsymb}$ 
where $S$ contains all states that appear in $\exec$. Additionally, we can make $\exec_-$ start with a write action (there is a transition in $\exec$ that writes $\asymb$). 
To obtain the decomposition, $\exec_+$ mimics $\exec$ but does not empty any state, 
and $\exec_-$ mimics $\exec$ but from a configuration with more states.   
We compute the maximum such set $S$ by iteratively deleting states that cannot appear in any "synchronizing execution". 
Let \begin{align*}
\covset{\prot} := & \max \set{S \subseteq Q \mid \exists \asymb \in \dataalp, \, \exists \aconfig_0 \in \aconfiginitset,  \, \aconfig_0 \step{*} \confpair{S}{\asymb}}
\\
 \cocovset{\prot} := & 
\max \set{ S \subseteq Q \mid \forall \asymb \in \dataalp, \, \exists \finalsymb \in \dataalp, \, \confpair{S}{\asymb} \step{*} \confpair{\targetstate}{\finalsymb}} \end{align*}
Both maxima exist as the sets are non-empty ($Q_0$ is included in the first set and $q_f$ is in the second set) and they are stable by union (concatenate the corresponding executions). 
Intuitively, $\covset{\prot}$ corresponds to the set of coverable sets, and $\cocovset{\prot}$ to the set of backward coverable states. In the decomposition $\exec_{+}: \aconfig_0 \step{*} \confpair{S}{\asymb}$, $\exec_-: \confpair{S}{\asymb} \step{*} \confpair{\targetstate}{\finalsymb}$, $\exec_+$ is a witness that $S \subseteq \covset{\prot}$ and $\exec_-$ that $S \subseteq \cocovset{\prot}$ (because $\exec_-$ starts with a write action, for every $\asymb' \in \dataalp$ one has $\confpair{S}{\asymb'} \step{*} \confpair{\targetstate}{\finalsymb}$). 

$\covset{\prot}$ and $\cocovset{\prot}$ can be computed in polynomial time. For $\covset{\prot}$, we use a "saturation" technique. For $\cocovset{\prot}$, we work backwards: a symbol is read before it is written. We start with $S := \set{\targetstate}$. Until a fixpoint for $S$ is reached, we do the following. We iterate on $\dataalp$, trying to pick the symbol that was in the register before $S$ could be reached. For each $\asymb \in \dataalp$, we saturate $S$ with backward transitions reading $\asymb$, then check if $\asymb$ can be written by a transition ending in $S$. If not, we backtrack by removing states that were just added.  

The algorithm iteratively removes from $\prot$ states that are not in $\covset{\prot} \cap \cocovset{\prot}$. Indeed, states that are not in $\covset{\prot} \cap \cocovset{\prot}$ cannot appear in any "synchronizing execution". If it ends up with $\states(\prot) = \emptyset$, then there is no "synchronizing execution" and the algorithm rejects. If it ends up with $\covset{\prot} = \cocovset{\prot} = \states(\prot) \ne \emptyset$, then applying the definitions of $\covset{\prot}$ and $\cocovset{\prot}$ gives a synchronizing execution, and the algorithm accepts.
\end{proof}

It is unknown whether the previous result still holds when $\rdim$ is fixed to a value greater than $1$. The case $\rdim = 1$ is particularly easy because writing to the register completely erases its content.

Unlike \COVER, \TARGET and therefore \dnfPRP are not tractable under the "uninitialized" hypothesis. For \TARGET, one cannot add fresh processes at no cost, since the fresh processes would eventually have to get to $\targetstate$. For example, if a register $\regvar$ can only be written from a given state $q$, the last process to leave $q$ will fix the value in register $\regvar$.

\begin{restatable}{proposition}{targetnpharduninit}
\label{prop:target_uninit_nphard_roundless}
\TARGET for "uninitialized" "roundless register protocols" is \NP-hard.
\end{restatable}


\begin{figure}[h]
\small
\begin{center}
\begin{tabular}{c || c | c | c | c |} 
 & \COVER & \TARGET & \dnfPRP & \PRP \\
\hhline{|=#=|=|=|=|}
 \multirow{2}{*}{General case} & \cellcolor{black!8}\NP-complete & \cellcolor{black!8}\NP-complete & \cellcolor{black!8} \NP-complete & \cellcolor{black!8} \NP-complete \\
& \cellcolor{black!8}(Prop.~\ref{prop:prp_npc_roundless} ) & \cellcolor{black!8}(Prop.~\ref{prop:prp_npc_roundless} ) & \cellcolor{black!8}(Prop.~\ref{prop:prp_npc_roundless} ) & \cellcolor{black!8}(Prop.~\ref{prop:prp_npc_roundless})\\
\hline
 \multirow{2}{*}{"Uninitialized"} & \PTIME-complete & \cellcolor{black!8} \NP-complete & \cellcolor{black!8} \NP-complete & \cellcolor{black!8} \NP-complete \\
&(Prop.~\ref{prop:ptime_cover_restrictions})
 & \cellcolor{black!8}(Prop.~\ref{prop:prp_npc_roundless} \& \ref{prop:target_uninit_nphard_roundless})
 & \cellcolor{black!8}(Prop.~\ref{prop:prp_npc_roundless} \& \ref{prop:target_uninit_nphard_roundless})
 & \cellcolor{black!8}(Prop.~\ref{prop:prp_npc_roundless} \& \ref{prop:prp_np_hard_restrictions})
\\
 \hline
 \multirow{2}{*}{$\rdim = 1$ (one register)} 
& \PTIME-complete 
& \PTIME-complete 
& \PTIME-complete 
& \cellcolor{black!8} \NP-complete \\
&(Prop.~\ref{prop:ptime_cover_restrictions}) 
& (Prop.~\ref{prop:ptime_dnfprp_onereg_roundless} \& \ref{prop:ptime_cover_restrictions}) 
& (Prop.~\ref{prop:ptime_dnfprp_onereg_roundless} \& \ref{prop:ptime_cover_restrictions}) 
& \cellcolor{black!8}(Prop.~\ref{prop:prp_npc_roundless} \& \ref{prop:prp_np_hard_restrictions}) \\
 \hline
 \end{tabular}
\end{center}
\caption{Summary of complexity results for "roundless register protocols"}
\label{fig:table_summary_complexity}
\end{figure}

\end{scope}
\begin{scope}

\section{Round-based Register Protocols}
\label{sec:setting_roundbased}

We now extend the previous model to a round-based setting. The model and semantics are the same as in \cite{myicalppaper}, however we consider a more general problem than \COVER. Thus, the abstract semantics developed here differs from \cite{myicalppaper}. 

\subsection{Definitions}
\label{subsec:defs_roundbased}
In round-based settings, there is a fresh set of $\rdim$ registers at each round, and each process has its own private round value that starts at $0$ and never decreases. Processes may only read from and write to registers of nearby rounds. 

\begin{definition}[Round-based register protocols]
  A ""round-based register protocol"" is a tuple $\prot =
    \tuple{\states, \initialstates,\rdim, \dataalp, \datainit, \vrange, \transitions}$
  where
  \begin{itemize}
    \item $\states$ is a finite set of states with a distinguished subset of 
          initial states~$\initialstates \subseteq \states$;
    \item $\rdim \in \nats$ is the number of shared registers per round; 
    \item $\dataalp$ is a finite data
          alphabet with an initial symbol $\datainit$;
    \item $\vrange$ is the visibility range;
    \item
          $\transitions \subseteq \states \times \actions \times \states$ is
          the set of transitions, where
          $\actions =  \{\readact{-i}{\regid}{\asymb} \mid i \in \zset{\vrange}, \allowbreak \regid \in \regint,\allowbreak \asymb \in \dataalp\} \cup \{
            \writeact{\regid}{\asymb} \mid \regid \in \regint, \asymb \in \datawrite\} \cup \set{\incr}$
          is the set of actions.
  \end{itemize}
\end{definition}

Read actions specify the round of the register: $\readact{-i}{\regid}{\asymb}$ means, for a process at round $k$, ``read $\asymb$ from register $\regid$ of round $k{-}i$''. 
A process at round $k$ may only write to the registers of round $k$. The $\incr$ action increments the round of a process. 

Let $\reg{k}{\regid}$ denote the register $\regid$ of round $k$. The set of registers of round $k$ is written $\regset{k}$, and we let $\regset{} = \bigcup_{k \in \NN} \regset{k}$. The size of a protocol is $|\prot| = |\states|+ |\dataalp| + |\transitions| + \vrange + \rdim$. 
A given process is described by its state and round, formalized by a pair $(q,k) \in \states \times \nats$ called ""location"". Let $\locations := \states \times \nats$ denote the set of locations. A ""concrete configuration@@roundbased"" describes the number of processes in each location along with the value of each register. 
Formally, a ""concrete configuration"" is a pair $\confpair{\mu}{\dvec}$ with $\mu \in \nats^{\locations}$ such that $0<\sum_{(q,k) \in \locations} \mu(q,k) < \infty$ and $\dvec \in \dataalp^{\regset{}}$. For $\cconfig = \confpair{\mu}{\dvec}$, we write $\loc{\cconfig} := \mu$ and $\datafun{\cconfig} := \dvec$. Again, we write $\cconfigs$ for the set of "concrete configurations". The set of initial configurations is  $\cconfiginitset := \set{\cconfig \in \cconfigs  \mid \datafun{\cconfig} = \datainit^{\regset{}} \text{ and } \forall (q,k) \notin \initialstates \times \set{0}, \, \loc{\cconfig}(q,k) = 0}$.  

A ""move"" is a pair $\move \in \transitions \times \nats$: move $(\delta, k)$ expresses that transition $\delta$ is taken by a process at round $k$; we write $\moves := \transitions \times \nats$ for the set of all moves. A "move" $\move$ has ""effect"" on round $k$ when $\move$ is at round $k$ or $\move$ is an increment at round $k{-}1$.
We define a "step" as follows: for $\move = ((q,a,q'),k) \in \moves$, $\cconfig \step{\move} \cconfig'$ when $(q,k) \in \loc{\cconfig}$ and:
\begin{itemize}
\item if $a = \readact{-i}{\regid}{\asymb}$, $\loc{\cconfig'} = (\loc{\cconfig} \ominus \set{(q,k)}) \oplus \set{(q',k)}$, $\data{\cconfig}{\reg{k{-}i}{\regid}} = \asymb$ and $\datafun{\cconfig'} = \datafun{\cconfig}$;
\item if $a = \writeact{\regid}{\asymb}$, $\loc{\cconfig'} = (\loc{\cconfig} \ominus \set{(q,k)}) \oplus \set{(q',k)}$, $\data{\cconfig'}{\reg{k}{\regid}} = \asymb$ and for all $\regvar \ne \reg{k}{\regid}$, $\data{\cconfig'}{\regvar} = \data{\cconfig}{\regvar}$;
\item if $a = \incr$, $\loc{\cconfig'} = (\loc{\cconfig} \ominus \set{(q,k)}) \oplus \set{(q',k{+}1)}$ and $\datafun{\cconfig'} = \datafun{\cconfig}$.
\end{itemize}

A step is ""at round $k$"" when the corresponding "move" is of the form $(\atrans, k)$. Note that action $\readact{-i}{\regid}{\asymb}$ is only possible for processes at rounds $k \geq i$. 
The notions of execution, of reachability and of coverability are defined as in the roundless case. 

\begin{example}
\label{example:roundbased_prot}
\begin{figure}
\centering
\resizebox{\linewidth}{!}{
\begin{tikzpicture}[x = 1cm, y = 1cm, node distance = 1.35cm]
\tikzstyle{every node} = [font = \footnotesize]
\node[state] (q0) at (0,0) {$q_0$};
\node[state, left = of q0] (A) {$A$};
\node[state, left = of A] (B) {$B$};
\node[state, left = of B] (C) {$C$};
\node[state, right = of q0] (D) {$D$};
\node[state, right = of D] (E) {$E$};
\node[state, right = of E] (qf) {$\targetstate$};

\path[-latex', thick]
(q0.270) edge[latex'-] +(270:4mm)
(q0) edge[loop above] node[above] {$\incr$} ()
(q0) edge node[below] {$\writeact{}{\exa}$} (A)
(A) edge node [below] {$\readact{-1}{}{\datainit}$} (B)
(B) edge node [below] {$\readact{-1}{}{\exa}$} (C) 
(C) edge[bend left = 20] node[above] {$\writeact{}{\exb}$} (q0)
(q0) edge node[below] {$\readact{-1}{}{\exb}$} (D)
(D) edge node[below] {$\readact{0}{}{\datainit}$} (E)
(E) edge node[below] {$\readact{0}{}{\exb}$} (qf);
\end{tikzpicture}
}
\caption{An example of round-based register protocol}
\label{fig:example_prot_roundbased}
\end{figure}
Consider the "round-based protocol" $\prot$ from Figure~\ref{fig:example_prot_roundbased}, with $\rdim = 1$, $\vrange = 1$, $\initialstates = \set{q_0}$ and $\dataalp = \set{\datainit, \exa, \exb}$. 
In this protocol, state $\targetstate$ cannot be covered.
By contradiction, consider an execution $\cexec: \cconfig_0 \step{*} \cconfig$ with $\cconfig_0 \in \cconfiginitset$ and $\loc{\cconfig}(\targetstate,k)>0$ fo some $k \in \nats$. We have that, at some point in $\cexec$, $(E,k)$ is populated and $\exb$ is in $\reg{}{k}$. Therefore, some process went from $(A,k)$ to $(B,k)$, which implies that $\reg{}{k}$ lost value $\datainit$ before $\reg{}{k{-}1}$; this in turn implies that $\cexec$ does not send any process to $(E,k)$ which is a contradiction. 
\end{example}

Since "round-based register protocols" enjoy the same monotonicity properties as "roundless register protocols", we define the same non-counting abstraction. Note that this abstraction differs from the one in \cite{myicalppaper} which was designed specifically for \COVER.  
The set of ""abstract configurations"" is $\aconfigs := 2^\locations \times \dataalp^{\regset{}}$; the abstract semantics are defined as in Subsection~\ref{subsec:abstract_semantics_roundless}. Again, $
\aconfig \step{\atrans} \aconfig$ if and only if there exist $\cconfig, \cconfig' \in \cconfigs, \, \cconfig \step{\atrans} \nobreak \cconfig'$ and $\abstractproj{\cconfig} = \aconfig, \abstractproj{\cconfig'} = \aconfig'.$
All the properties of Subsection~\ref{subsec:abstract_semantics_roundless} apply to round-based abstract semantics. In particular, we have the soundness and completeness of the abstraction:

\begin{restatable}[Soundness and completeness of the abstraction]{proposition}{soundcompleteroundbased}
  \label{prop:soundcomplete_roundbased}
  For all $L \subseteq \locations$, $\dvec \in \dataalp^{\regset{}}$:
  \[
    (\exists \cconfig \in \creach{\cconfiginitset{}}:
    \supp{\cconfig} {=} L, \, \datafun{\cconfig} {=} \dvec)
     \  \Longleftrightarrow\ \\ (\exists \aconfig \in \areach{\aconfiginitset{}}:\
    \loc{\aconfig} {=} L, \, \datafun{\aconfig} {=} \dvec). 
  \]
\end{restatable}

\subsection{Presence Reachability Problem}
\label{subsec:def_prp_roundbased}
\COVER is extended to round-based protocols by asking whether some reachable configuration has a process on $\errorstate$ on some round $k$, and \TARGET by asking whether some reachable configuration has no process on states $q \ne \targetstate$ on any round $k$. Formally, one asks whether there exists $\cconfig \in \creach{\cconfiginitset}$ such that $\cconfig \models \psi$ where $\psi=$``$\exists k \in \NN, (q,k) \in \loc{\cconfig}$'' for \COVER and $\psi=$``$\forall k \in \nats, \forall q \ne \targetstate, (q,k) \notin \loc{\cconfig}$'' for \TARGET. 
We will now extend roundless \PRP to round-based \PRP, where the formula is allowed to have non-nested quantification over rounds.  

"Presence constraints" are first-order formulas (quantifying over the rounds)  without any nested quantifiers. See Appendix~\ref{def:presence_constraints_roundbased} for the full definition. 


 

\begin{example}
``$(\exists k \, (q_2,k) \populated) \lor (\forall k \, ((q_0,k{+}2) \populated) \land \reg{1}{1} \contains \exa)$'' is an example of "presence constraint". \\ 
Let $\cconfig := ((q_0,0) \oplus (q_1,1), \datainit^{\regset{}})$ with $q_0 \ne q_1$, $\rdim = 1$. One has $\cconfig \models (\reg{0}{1} \contains \nobreak \datainit) \land (\exists k \, (q_1,k{+}1) \populated)$ but  $\cconfig \not \models \forall k \, (((q_0,k) \populated) \lor \neg ((q_1,k) \populated))$. 
\end{example}

We define the ""round-based presence reachability problem"" (round-based \PRP): 
\smallskip

\noindent\fbox{\begin{minipage}{.99\linewidth}
    \textsc{Round-based \PRP} \\
    {\bf Input}: A round-based register protocol 
    $\prot$, a "presence constraint" $\psi$\\
    {\bf Question}: Does there exist $\cconfig \in \creach{\cconfiginitset}$ such that $\cconfig \models \psi$?
  \end{minipage}} \smallskip

\begin{example}
Consider $\prot$ from Example~\ref{example:roundbased_prot}.
If $\psi := \exists k, (\targetstate,k) \populated$, then $(\prot, \psi)$ is a negative instance of round-based \PRP. If $\psi' := \exists k, ((E,k) \populated) \allowbreak \land ((E,k{+}1) \populated)$, then $(\prot, \psi')$ is also negative. However, if $\psi'' := ((E,2) \populated) \allowbreak \land [\forall k, (\reg{}{k+1} \contains \exb) \lor (\reg{}{k+1} \contains \datainit)]$, then $(\prot, \psi'')$ is positive: a witness execution sends a process to $(B,1)$, writes $\exa$ to $\reg{}{0}$ then $\exb$ to $\reg{}{1}$ and finally sends a process from $(q_0,2)$ to $(E,2)$. 
\end{example}

\COVER and \TARGET for "round-based register protocols" are special cases of \PRP. The following lower bound hence applies to all these problems:

\begin{proposition}[{\cite[Theorem 23]{myicalppaper}}]
\label{prop:pspace_hard_cover_roundbased}
\COVER for "round-based register protocols" is \PSPACE-hard, even in the uninitialized case with $\vrange =0$ and $\rdim = 1$. 
\end{proposition}
Note that, in the round-based setting, $\rdim = 1$ means one register \emph{per round}, therefore still an unbounded number of registers. $\vrange =0$ means that a process can only interact with registers of its current round.
The previous proposition implies that all problems considered in Figure~\ref{fig:table_summary_complexity} are \PSPACE-hard when working with round-based protocols. In \cite{myicalppaper}, \COVER for "round-based register protocols" is shown to be \PSPACE-complete. In the rest of this paper, we establish that the more general "round-based \PRP" lies in the same complexity class:

\begin{theorem}
\label{thm:prp_pspace_roundbased}
"Round-based \PRP" is \PSPACE-complete.
\end{theorem}

\section{A Polynomial-Space Algorithm for Round-Based \PRP}
\label{sec:algo_pspace_prp_roundbased}
In this section, we provide a polynomial-space algorithm for round-based \PRP. Thanks to Savitch's theorem, it suffices to find a non-deterministic polynomial-space algorithm. 
To do so, one wants to guess an execution that reaches a configuration satisfying the "presence constraint". However, as shown in \cite[Proposition 13]{myicalppaper}, one may need, at a given point along such an execution, the number of active rounds to be exponential (an active round being informally a round on which something has already happened and something else is yet to happen). Thus, storing the execution step by step in polynomial space seems hard; instead, our algorithm will guess the execution round by round. To do this, we define the notion of "footprint", which represents the projection of an execution onto a narrow window of rounds.

Thanks to Proposition~\ref{prop:soundcomplete_roundless}, round-based \PRP can be studied directly in the abstraction. In the rest of the paper, all configurations and executions are implicitly abstract. 

\subsection{Footprints}
\label{subsec:footprints}
Let $j \leq k$. We write $\locallocations{j}{k}$ for the set of locations at rounds $\max(j,0)$ to $k$; similarly, we write $\localregisters{j}{k}$ for the set of registers of rounds $\max(j,0)$ to $k$.  
A~""local configuration"" on (rounds) $\iset{j}{k}$ is an element of $2^{\locallocations{j}{k}} \times \dataalp^{\localregisters{j}{k}}$. The set of "local configurations" on $\iset{j}{k}$ is written $\localconfigurations{j}{k}$. 
Given $\aconfig \in \aconfigs$, the "local configuration" 
$\localconfigproj{j}{k}{\aconfig}$
is obtained by removing from $\aconfig$ all information that is not about rounds $j$ to $k$. Note that "local configurations" are local with respect to the rounds, and not with respect to processes.

Given $\localconfig, \localconfig' \in \localconfigurations{j}{k}$ and a "move" $\move$, we write $\localconfig \step{\move} \localconfig'$ when there exist two configurations $\aconfig$ and $\aconfig'$ such that $\aconfig \step{\move} \aconfig'$, $\localconfigproj{j}{k}{\aconfig} = \localconfig$ and $\localconfigproj{j}{k}{\aconfig'} = \localconfig'$. In practice:
\begin{itemize}
\item if $\move$ is a move with no "effect" on rounds $j$ to $k$, then
$\localconfig \step{\move} \localconfig'$ if $\localconfig = \localconfig'$;
\item if $\move = ((q, \incr,q'),j{-}1)$ then $\localconfig \step{\move} \localconfig'$ holds with no condition that $(q,j{-}1)$ is populated in $\localconfig$, since $j{-}1$ is outside of $\iset{j}{k}$; 
\item if $\move = ((q, \readact{-b}{\regid}{\asymb}), l)$ with $l{-}b < j$ (read from register of round $<j$), there is no condition on the content of the register.
\end{itemize}

A ""footprint"" on (rounds) $\iset{j}{k}$ corresponds to the projection of an execution on rounds $\iset{j}{k}$. Formally, it is an alternating sequence $\localconfig_0, \move_0, \localconfig_1, \dots, \move_{m-1}, \localconfig_m$ where for all $i \in \zset{m}$, $\localconfig_i \in \localconfigurations{j}{k}$ and for all $i \leq m-1$, $\localconfig_i \step{\move_i} \localconfig_{i+1}$ and $\localconfig_i \ne \localconfig_{i+1}$.  

Let $\exec = \aconfig_0, \move_0, \aconfig_1, \dots, \move_{m{-}1}, \aconfig_{m}$ be an execution. The ""footprint of $\exec$ on (rounds) $\iset{j}{k}$"", written $\fpproj{j}{k}{\exec}$,  is the "footprint" on $\iset{j}{k}$ 
obtained from $\exec$ by replacing $\aconfig_i$ by $\localconfig_i = \localconfigproj{j}{k}{\aconfig_i}$ and then removing all useless steps $\localconfig_{i} \step{\move} \localconfig_{i{+}1}$ with $\localconfig_i = \localconfig_{i{+}1}$ (by merging $\localconfig_i$ and $\localconfig_{i{+}1}$, so $\fpproj{j}{k}{\exec}$ can be shorter than $\exec$).  Similarly, for $\iset{j'}{k'} \supseteq \iset{j}{k}$ and $\afootprint$ a "footprint" on $\iset{j'}{k'}$, define the ""projection@@footprint"" $\fpproj{j}{k}{\afootprint}$ by the footprint obtained by replacing each "local configuration" in $\afootprint$ by its projection on $\iset{j}{k}$ and removing useless steps.

The following result provides a sufficient condition for a sequence of footprints to be seen as projections of a single common execution.
\begin{restatable}{lemma}{combiningfootprints}
\label{lem:combiningfootprints}
Let  $K \in \NN$, $(\afootprint_k)_{k \leq K}$ and $(T_k)_{k \leq K{-}1}$ such that:
\begin{itemize}
\item for all $k \leq K$, $\afootprint_k$ is a "footprint" on $\iset{k{-}\vrange{+}1}{k}$,
\item for all $k \leq K{-}1$, $T_k$ is a "footprint" on $\iset{k{-}\vrange{+}1}{k{+}1}$,
\item for all $k \leq K{-}1$, $\fpproj{k{-}\vrange{+}1}{k}{T_{k}} = \afootprint_{k}$,
\item for all $k \leq K{-}1$, $\fpproj{k{-}\vrange{+}2}{k{+}1}{T_{k}} = \afootprint_{k{+}1}$.
\end{itemize}
 There exists an execution $\exec$ such that, for all $k \leq K$, $\fpproj{k{-}\vrange{+}1}{k}{\exec} = \afootprint_k$. 
\end{restatable}

\subsection{A Polynomial-Space Algorithm for Round-Based \PRP}
\label{subsec:outline_algorithm}

The algorithms guesses the witness execution "footprint" by "footprint", and stops when the "presence constraint" is satisfied. 
Algorithm~\ref{algo:pspace} provides the skeleton of this procedure. For the sake of simplicity, we suppose that $\vrange \geq 1$. If $\vrange = 0$, we artificially increase $\vrange$ to $1$.

\begin{algorithm}[htbp]
\SetKwFunction{InitFun}{\initalgo}
\SetKwFunction{OnestepFun}{\onestepalgo}
\SetKwFunction{TestFun}{\testalgo}
\textbf{Input:} A \PRP instance $(\prot, \psi)$ \\
  $E, U, C \assignalgo \emptyset$ \;
  $\tau \assignalgo \epsilon$ \nllabel{line_initialisation}  
   \tcp*[r]{dummy "footprint" on rounds $\iset{-\vrange}{-1}$}
  Guess the initial set $I \subseteq \initialstates$ of populated states at round $0$  \nllabel{line_guess_initial} \; 
  \InitFun{$E,U, C$} \nllabel{line_initfun} \;
  \For{$k$ from $0$ to
    $+ \infty$ \nllabel{line_loop_round}}{
    Guess $T$ a "footprint" on $\iset{k{-}\vrange}{k}$ such that $\fpproj{k{-}\vrange}{k{-}1}{T} = \tau$ \;
    Check that $T$ is consistent with the initial configuration \nllabel{line_check_initial} \;
    $\localconfig \assignalgo $ last configuration in $T$ \; 
    \OnestepFun{$E, U, C, \localconfig$} \nllabel{line_onestepfun}\;
    \lIf{\TestFun{$E, U, C, \localconfig$} \nllabel{line_testfun}}{
      Accept \nllabel{line_accept}
    }
    $\tau \assignalgo \fpproj{k{-}\vrange{+}1}{k}{T}$ \;
    }
    
\caption{Non-deterministic algorithm for round-based \PRP}\label{algo:pspace}
\end{algorithm}

For all $k \in \NN$, let $\tau_{k}$ be the value of $\tau$ at the end of iteration $k$ and $T_k$ the value of $T$ guessed at iteration $k{+}1$. Thanks to Lemma~\ref{lem:combiningfootprints}, if the algorithm reaches the end of iteration $K$ then there exists an execution $\exec$ whose projection on $\iset{k{-}\vrange}{k-1}$ is $\tau_k$ for every $k \leq K$.  

Handling the "round-based presence constraint" is technical, so we hide it in functions \texttt{\initalgo}, \texttt{\onestepalgo} and \texttt{\testalgo} and postpone the details to Appendix~\ref{appendix:details_algo_pspace}. We guess why $\psi$ is true by guessing satisfied \emph{atomic propositions} of three types: existentially quantified on the round (\emph{i.e.}, of the form ``$\exists k \, \phi$'' where $\phi$ has no quantifiers and only $k$ as free variable) which we put in $E$; universally quantified on the round (\emph{i.e.}, of the form ``$\forall k \, \phi$'' where $\phi$ has no quantifiers and only $k$ as free variable) which we put in $U$; with no quantifier (\emph{i.e.}, of the form ``$\phi$'' where $\phi$ has no quantifiers and no free variables) which we put in $C$. 
Formulas in $C$ refer to constant rounds and are checked at these rounds only.  Formulas in $U$ are checked at every round. For formulas in $E$, the algorithm guesses at which round the formula is true.
Our algorithm is correct with respect to round-based \PRP:

\begin{restatable}{proposition}{correctnessalgo}
\label{prop:correctness_algo_pspace}
$(\prot, \psi)$ is a positive instance of round-based \PRP if and only if there exists an accepting computation of Algorithm~\ref{algo:pspace} on $(\prot, \psi)$.
\end{restatable}

The integer constants in the "presence constraint" $\psi$ are encoded in unary, like the "visibility range" $\vrange$. These two hypotheses are reasonable since practical examples typically use constants of small value (\emph{e.g.}, $1$). Under these hypotheses, we obtain a polynomial spatial bound on the size of "footprints" of a well-chosen witness execution, which in turn gives a polynomial spatial bound for the algorithm: 
\begin{restatable}{proposition}{algopolyspace}
\label{prop:polynomial_space}
Algorithm~\ref{algo:pspace} works in space $O(|\psi|^3 + |Q|^2 \, (\vrange{+}1)^2 \, \log(\rdim \, |\dataalp|))$.
\end{restatable}

Finally, we need to discuss the termination of the algorithm.
According to the pigeonhole principle, after an exponential number of iterations, the elements stored in memory repeat from a previous iteration and we can stop the computation. 
One can thus use a counter, encoded in polynomial
space, to count iterations and return a
decision when the counter reaches its largest value.
Thanks to the space bounds from Proposition~\ref{prop:polynomial_space}, correctness from Proposition~\ref{prop:correctness_algo_pspace} and the stopping criterion, our algorithm decides "round-based \PRP" in non-deterministic polynomial space, proving Theorem~\ref{thm:prp_pspace_roundbased}.  

\paragraph*{Acknowledgements}
Many thanks to Nathalie Bertrand, Nicolas Markey and Ocan Sankur for their invaluable advice.
\end{scope}

\bibliography{refs}

\appendix
\newpage
\pagebreak
\setcounter{theorem}{0}
\def\theHtheorem{\theHsection.\arabic{theorem}}
\def\theHlemma{\theHsection.\arabic{lemma}}
\def\theHproposition{\theHsection.\arabic{proposition}}
\def\theHcorollary{\theHsection.\arabic{corollary}}
\def\theHexample{\theHsection.\arabic{example}}
\def\thetheorem{\thesection.\arabic{theorem}}
\def\thelemma{\thesection.\arabic{lemma}}
\def\theproposition{\thesection.\arabic{proposition}}
\def\thecorollary{\thesection.\arabic{corollary}}
\def\theexample{\thesection.\arabic{example}}

\noindent {\LARGE\bfseries\sffamily Technical appendix}

\section{Roundless Register Protocols}
Recall that the abstract semantics is non-deterministic due to the choice between $\state{\cconfig'} = (\state{\cconfig} \setminus \set{q}) \cup \set{q'}$ and  $\state{\cconfig'} = \state{\cconfig} \cup \set{q'}$. The first case is called ""deserting"", and the latter ""non-deserting"". A "deserting" "step" corresponds in the concrete semantics to all processes in $q$ taking the transition at once.

Moreover, given a transition $(q,a,q')$, state $q$ is called the ""source"" and state $q'$ is called the ""destination"" of the "transition", and similarly for a step.

\subsection{Proof of Proposition~\ref{lem:copycat_roundless}}
\label{appendix:proof_copycat}
\copycat*

\begin{proof}
We prove the result by induction on the length of the execution.
If the execution is of length $0$ then one simply considers $q_1 := q_2$.
Let $\cexec: \cconfig_1 \step{*} \cconfig_2$ and suppose that the property is true for all executions of length $|\cexec| -1$. Decompose $\cexec$ into $\cconfig_1 \step{*} \cconfig_3 \step{\atrans} \cconfig_2$. If $q_2$ is not the destination of $\atrans$, then $q_2 \in \supp{\cconfig_3}$ and we directly apply the induction hypothesis on $\cconfig_1 \step{*} \cconfig_3$ and $q_2$ then conclude by taking $\atrans$ to get to $\confpair{\state{\cconfig_2} \oplus q_2}{\datafun{\cconfig_2}}$.
Assume that $q_2$ is the destination of $\atrans$; let $q_3$ the source of $\atrans$. We have $q_3 \in \supp{\cconfig_3}$, so we apply the induction hypothesis on $\cconfig_1 \step{*} \cconfig_3$ and $q_3$: we obtain that there exists $q_1 \in \supp{\cconfig_1}$ such that $\confpair{\state{\cconfig_1} \oplus q_1}{\datafun{\cconfig_3}} \step{*} \confpair{\state{\cconfig_3} \oplus q_3}{\datafun{\cconfig_2}}$ Moreover, we have $\confpair{\state{\cconfig_3} \oplus q_3}{\datafun{\cconfig_2}} \step{\atrans}  \confpair{\state{\cconfig_2} \oplus q_3}{\datafun{\cconfig_2}} \step{\atrans}  \confpair{\state{\cconfig_2} \oplus q_2}{\datafun{\cconfig_2}}$. Indeed, if $\atrans$ is a read transition then the symbol is still in the register in $\cconfig_2$ and may be read again, and if $\atrans$ is write transition, then writing again a symbol to a register does not change its content. 
\end{proof}

\subsection{Proof of Proposition~\ref{prop:soundcomplete_roundless}}
\label{appendix:proof_soundcomplete}
\soundcompleteroundless*
First, the following lemma states that any concrete execution can easily be transformed into an abstract one. 
\begin{lemma}
\label{lem:concrete_to_abstract}
Let $\cconfig, \cconfig' \in \cconfigs$ and $\cexec: \cconfig \step{*} \cconfig'$. There exists $\exec: \abstractproj{\cconfig} \step{*} \abstractproj{\cconfig'}$. 
\end{lemma}
\begin{proof}
By induction, it suffices to prove it for one step; suppose that $\cconfig \step{\atrans} \cconfig'$. 
Let $(q,a,q') := \atrans$; if $q \in \state{\cconfig'}$, then we consider the "non-deserting" abstract step with transition $\atrans$, otherwise we consider the "deserting" step with transition $\atrans$. Either way, we have 
 $\abstractproj{\cconfig} \step{\atrans} \abstractproj{\cconfig'}$.  
\end{proof}

Conversely, from an abstract execution, for a large
enough number of processes, using the copycat property one can build a
concrete execution with the same final states and data. 
\begin{lemma}
\label{lem:abstract_to_concrete}
Let $\aconfig, \aconfig' \in \aconfigs$ and $\exec : \aconfig \step{*} \aconfig'$. There exist $\cconfig, \cconfig'$ such that $\abstractproj{\cconfig} = \aconfig$, $\abstractproj{\cconfig'} = \aconfig'$ and $\cexec: \cconfig \step{*} \cconfig'$. 
\end{lemma}
\begin{proof}
The proof is by induction on the number of steps in $\cexec$. If $\cexec$ has $0$ steps, then $\aconfig = \aconfig'$ and suffices to consider $\cconfig = \cconfig' :=  \confpair{\bigoplus_{q \in \state{\aconfig}}q}{\datafun{\aconfig}}$.

Assume that $\exec: \aconfig_1 \step{*} \aconfig_2 \step{\atrans} \aconfig_3$ and that the property is true for executions shorter than $\exec$. By induction hypothesis, there exists $\cexec: \cconfig_1 \step{*} \cconfig_2$ such that $\abstractproj{\cconfig_1} = \aconfig_1$, $\abstractproj{\cconfig_2} = \aconfig_2$. There exists $\cconfig_3$ such that 
$\cconfig_2 \step{\atrans} \cconfig_3$; however, it could be that $\cconfig_3$ does not have a process on the "source" state of $\atrans$ while $\aconfig_3$ does. In that case, we modify $\cexec$ to put an additional process on the source "state" of $\atrans$ by using the copycat property and increasing the number of processes by one.  
\end{proof}

Lemmas~\ref{lem:concrete_to_abstract} and \ref{lem:abstract_to_concrete} together prove Proposition~\ref{prop:soundcomplete_roundless}.

\subsection{Proof of Proposition~\ref{prop:ptime_cover_restrictions}}
\coverrestrictions*
\subsubsection{\PTIME when the Registers are "Uninitialized"}
The algorithm computes the set of coverable states using a fixpoint technique called ""saturation"".
The algorithm starts with $S := \{\initialstates\}$, then iteratively adds to $S$ all states $q_2$ for which there exist $q_1 \in S$ and an action $a \in \actions$ such that $(q_1,a,q_2) \in \transitions$ 
and:
\begin{itemize}
\item either $a = \writeact{\regid}{\asymb}$ with $\asymb \in \dataalp$, 
\item or $a = \rlreadact{\regid}{\asymb}$ with $\regid, \asymb$ s.t. there exist $q_3, q_4 \in S$, $(q_3, \writeact{\regid}{\asymb}, q_4) \in \transitions$.
\end{itemize}
We prove that, when a fixpoint is reached, $S$ is exactly the set of coverable states.
First, any coverable state is added to $S$ by the algorithm, by induction in the number of steps of the execution. Conversely, we show by induction in the number of iterations of the algorithm that any state added to $S$ is coverable. For the first case, if $(q_1, \writeact{\regid}{\asymb}, q_2) \in \transitions$ and $q_1$ is coverable then clearly $q_2$ is coverable.
For the second case, suppose that we have an execution $\exec$ covering $S \subseteq \states$, and that there exist $q_2 \in Q, q_1, q_3, q_4 \in S$, $\asymb \in \datawrite$, $\regid \in \nset{\rdim}$ such that $(q_1,\rlreadact{\regid}{\asymb},q_2), (q_3, \writeact{\regid}{\asymb}, q_4) \in \transitions$. Because we have an unlimited supply of processes, we use the copycat property to put an extra process on $q_3$ then make that process write $\asymb$ to $\rlreg{\regid}$ again, so that a process in state $q_1$ may read $\asymb$ from register $\rlreg{\regid}$ and get to $q_2$, which is therefore coverable.


\subsubsection{\PTIME when the Number of Registers $\rdim$ is Fixed}
The first write to a register is an irreversible action, as $\datainit$ cannot be written again. For that reason, we cannot work with a single saturation phase like in the "uninitialized" case. We iterate over all possible orders in which registers are first written to (see \emph{first-write orders} in \cite[Definition 15]{myicalppaper}).

For a given such order $\regvar_1, \dots, \regvar_m$ ($m \leq \rdim$), we proceed using $m+1$ successive saturation phases, numbered from $i=0$ to $m$.
The algorithm starts with $S = \set{\initialstates}$. During saturation phase $i$, the algorithm saturates $S$ by iteratively adding all states $q_2$ such that:

\begin{itemize} 
\item there exists $q_1 \in S$ with $(q_1, \rlreadact{\regvar}{\datainit}, q_2)$ and $\regvar \notin \set{\regvar_1,\dots, \regvar_i}$,
\item there exists $q_1 \in S$ with $(q_1, \writeact{\regvar_j}{\asymb}, q_2)$, $j\leq i$,
\item there exists $q_1 \in S$ with $(q_1, \rlreadact{\regvar_j}{\asymb}, q_2)$, $j\leq i$, and $\asymb$ may be written to $\regvar_j$ using a transition whose source is in $S$.
\end{itemize}

First, if there exists an execution covering $\targetstate$, then it writes to registers for the first time in some order $\regvar_1, \dots, \regvar_m$. When the algorithm considers this first-write order, the set of states computed includes $q_f$.  Conversely, suppose that the algorithm finds that $q_f$ is covered for some first-write order $\regvar_1, \dots, \regvar_m$. Observe that, if two executions share a first-write order $\regvar_1, \dots, \regvar_m$ then they may be merged into a common execution \cite[Lemma 17]{myicalppaper}. Therefore, all the states computed by the algorithm may be covered in a single, big execution and $q_f$ is coverable. 

\subsubsection{\PTIME-hardness}
The proof is similar to the one presented in \cite[Proposition 1]{delzanno2012} for broadcast protocols. 
It uses a \LOGSPACE-reduction for the Circuit Value Problem, which is \PTIME-complete for \LOGSPACE reductions \cite{pcomplete}. This problem consists in determining the output value of an acyclic Boolean circuit with given input values and Boolean gates that can be negations $\neg$, disjunctions $\lor$ and conjunctions $\land$. 

Consider an instance of the Circuit Value Problem, we write $V$ for the set of input, intermediate and output values of the circuit. A gate is represented as a tuple of the form $(\neg, i, o)$, $(\lor, i_1, i_2, o)$ or $(\land, i_1, i_2, o)$ where $i, i_1, i_2 \in V$ denote input(s) and $o \in V$ the output of the gate.
We construct an instance $(\protcvp, \errorstate)$ of \COVER with $\rdim = 1$. In $\dataalp$, we have $\datainit$ (which is never read) along with, for every $v \in V$, symbols $\cvptrue{v}$ and $\cvpfalse{v}$, denoting that $v$ is respectively true and false. 
First, $\protcvp$ has a part containing a state in $\initialstates$ from which one may write the symbols corresponding to the assignment of the input values of the circuit. 
Moreover, for every gate of the circuit, there is a part of the protocol corresponding to this gate, which has a state in $\initialstates$ from which a process may read the values of the inputs and write the corresponding value of the output. A depiction for gate $(\land, i_1, i_2, o)$ may be found in Figure~\ref{fig:subprot_land_gate}. Lastly, state $\errorstate$ is the destination of the transition writing the symbol corresponding to the output variable of the circuit having the desired value, so that $\errorstate$ is coverable if and only if this desired value is indeed the output value of the circuit. 

\begin{figure}[h]
\begin{center}
\begin{tikzpicture}[node distance = 1cm, auto]
\tikzstyle{every node}=[font=\normalsize]
\tikzstyle{every state} = [minimum size = 0.5cm]

\node[state] (init) at (0,0) {};
\node[state] at (0,-2) (readT1) {};
\node[state] at (3,-2) (readT2) {};
\node[state] (readF) at (3,0) {};

\draw  (init.180) edge[latex'-] +(180:4mm);
\path[-latex, draw] 
(init) edge[bend left] node {$\rlreadact{}{\cvpfalse{i_1}}$} (readF)
(init) edge[bend right] node[below] {$\rlreadact{}{\cvpfalse{i_2}}$} (readF)
(readF) edge[loop right] node {$\writeact{}{\cvpfalse{o}}$} ()
(init) edge node[left] {$\rlreadact{}{\cvptrue{i_1}}$} (readT1)
(readT1) edge node {$\rlreadact{}{\cvptrue{i_2}}$} (readT2)
(readT2) edge[loop right] node {$\writeact{}{\cvptrue{o}}$} ();
\end{tikzpicture}
\end{center}
\caption{Part of the protocol $\protcvp$ that corresponds to gate $(\land, i_1, i_2, o)$}
\label{fig:subprot_land_gate}
\end{figure}

\subsection{Proof of Proposition~\ref{prop:ptime_dnfprp_onereg_roundless}}
\label{appendix:ptime_dnfprp_onereg_roundless}
\ptimednfprponereg*

We here prove the result in the general case of Proposition~\ref{prop:ptime_dnfprp_onereg_roundless}. We first prove that it suffices to prove the result for "uninitialized protocols".

\begin{lemma}
\label{lem:reduc_init_to_uninit}
There exists a polynomial-time reduction from "initialized" \dnfPRP with $\rdim =1$ to "uninitialized" \dnfPRP with $\rdim = 1$. 
\end{lemma}
\begin{proof}
Let $\prot = \tuple{\states, \initialstates, 1, \dataalp, \datainit, \transitions}$ a "roundless register protocol" with a single register.
Any execution will be composed of two phases: the phase where the register has value $\datainit$ and no write transition is taken and the phase where the register no longer has value $\datainit$ and write transitions may be taken. The reduction relies on the observation that, in the first phase, only transitions labeled by $\rlreadact{}{\datainit}$ may be taken, and processes do not interact during this phase. Therefore, we can consider as initial any state that may be covered from $\initialstates$ with a path of transitions all labeled by $\rlreadact{}{\datainit}$.

Consider the graph $G=(Q,E)$ whose vertices are the states of the system and whose edges are the transitions labeled by $\rlreadact{}{\datainit}$: $(q_1,q_2) \in E$ if and only if $(q_1, \rlreadact{}{\datainit},q_2) \in \transitions$. Let $\initialstates' := \set{q \in \states \mid \exists m \geq 0, \exists q_0 \in \initialstates, \, \exists q_1, \dots, q_{m-1}, q_m = q \in \states, \, \forall i \in \iset{0}{m-1}, \, (q_i,q_{i+1}) \in E}$ the set of states reachable from $\initialstates$ in $G$. $\initialstates'$ can trivially be computed in polynomial time. Additionally, let $\transitions' := \transitions \setminus  \set{(q,\rlreadact{}{\datainit}, q')\mid q, q' \in Q}$. The reduction maps $\prot$ to the protocol $\prot' = \tuple{\states, \initialstates', 1, \dataalp, \datainit, \transitions'}$.

We now prove that $(\prot, \psi)$ is a positive instance of \dnfPRP if and only if $(\prot',\psi)$ is. First, suppose that, in $\prot$, there exists $\exec: \aconfig_0 \step{*} \aconfig$ such that $\aconfig \models \psi$. We decompose $\exec$ into $\exec_p: \aconfig_0 \step{*} \aconfig_1$ and $\exec_s: \aconfig_1 \step{*} \aconfig$ where $\data{\aconfig_1}{\regvar} = \datainit$ and $\exec_s$ either is the empty execution or starts with a write transition. $\exec_p$ only uses transitions labeled by $\rlreadact{}{\datainit}$ therefore, for every $q \in \state{\aconfig_1}$, there exists a path in $G$ from $\initialstates$ to $q$; this proves that $\state{\aconfig_1} \subseteq \initialstates'$ and therefore that $\aconfig_1$ is initial for $\prot'$, moreover $\exec_s$ does not use transitions labeled by $\rlreadact{}{\datainit}$ hence $\exec_s$ is a witness that $(\prot', \psi)$ is positive.

Suppose now that $(\prot', \psi)$ is positive. There exists $\exec: \aconfig_0 \step{*} \aconfig$ with $\aconfig \models \psi$.  For every $q \in \state{\aconfig_0} \setminus \initialstates$, there exists $f(q) \in \initialstates$ such that $q$ is reachable from $f(q)$ in $G$. Let $S:= (\state{\aconfig_0} \cap \initialstates) \cup f(\state{\aconfig_0} \setminus \initialstates)$, we have $S \subseteq \initialstates$. We have that $\confpair{S}{\datainit} \step{*} \aconfig_0$ in $\prot$ by only taking transitions appearing in $G$. Therefore $\confpair{S}{\datainit} \step{*}  \aconfig$ with $\aconfig \models \psi$ and $\confpair{S}{\datainit} $ is initial for $\prot$, which proves that $(\prot, \psi)$ is positive. 
\end{proof}

Thanks to the previous lemma, we prove Proposition~\ref{prop:ptime_dnfprp_onereg_roundless} in the "uninitialized case". 
Consider a instance $(\prot, \psi)$ of \dnfPRP where $\prot$ is "uninitialized" and $\rdim = 1$. $(\prot, \psi)$ is positive if and only if $(\prot, \clause)$ is positive for some clause $\clause$ of $\psi$. Our algorithm hence iterates over all clauses in $\psi$.

Let $\clause$ a clause in $\psi$. $\clause$ is a conjunction of literals, hence it may be seen as a set of atomic propositions that the configuration reached has to satisfy. Let $Q_+(\clause)$ be the set of states that need to be populated in the final configuration, $Q_-(\clause)$ the states that need to not be populated, and $\dataok(\clause)$ the symbols that are allowed in the final configuration. Formally, $Q_+(\clause) := \set{q \mid \text{``$q \populated$''} \in \clause}$,
 $Q_-(\clause) := \set{ q \mid \text{``$\neg (q \populated)$''} \in \clause}$
 and $\dataok(\clause) := \set{\asymb \in \dataalp \mid 
\text{``$\neg(\regvar \contains \asymb)$''} \notin \clause 
\text{ and } \forall \asymb' \ne \asymb, \, \text{``$\regvar \contains \asymb'$''} \notin \clause}$ 
where $\regvar$ denotes the register. For all $\confpair{S}{\asymb} \in \aconfigs$, $\confpair{S}{\asymb} \models \clause$ if and only if $Q_+(\clause) \subseteq S \subseteq Q \setminus Q_-(\clause)$ and $\asymb \in \dataok(\clause)$. Let \[ \allowedsets(\clause) := \set{S \subseteq Q \mid Q_+(\clause) \subseteq S \subseteq Q \setminus Q_-(\clause)} \] denote the collection of all sets of states allowed in the final configuration. 

\begin{lemma}
\label{lem:decomp_exec_dnfprp_onereg}
Any execution $\exec: \aconfig_0 \step{*} \aconfig$ with $\aconfig_0 \in \aconfiginitset$ may be decomposed in the following form: $\aconfig_0 \step{*} \confpair{S}{\finalsymb} \step{*} \aconfig$ with $S$ containing all states appearing in $\exec$.
\end{lemma}
\begin{proof}
Let $\exec: \aconfig_0 \step{*} \aconfig$; let $\confpair{S_f}{\finalsymb} := \aconfig$.  
The execution $\aconfig_0 \step{*} \confpair{S}{\finalsymb}$ is obtained by turning into "non-deserting" all "deserting" steps in $\exec$, so that all states covered in $\exec$ appear in $S$.
 For the second execution, we claim that there exists  $\exec': \confpair{S}{\finalsymb} \step{*} \confpair{S_f'}{\finalsymb}$ that is obtained by mimicking steps of $\exec$ starting from $\confpair{S}{\finalsymb}$. 
First, $\prot$ is "uninitialized" therefore $\exec$ starts with a write and the register value at the beginning of $\exec$ is irrelevant. 
Moreover, $\state{\aconfig_0} \subseteq S$ by definition of $S$, so that $\exec'$ starts from a configuration with more states that $\exec$. By induction, for all $n \geq 1$, the $n$-th configuration in $\exec'$ has the same register value as and more states than the $n$-th configuration in $\exec$. This in fact proves that $S_f \subseteq S_f'$.   
Moreover, for every $q \in S \setminus S_f$, since $q \notin S_f$ the last step in $\exec$ about step $q$ has $q$ as source and is "deserting", hence $q$ is also deserted in $\exec'$ which shows that $S_f' \subseteq S_f$. In the end, $\exec'$ goes from $\confpair{S}{\finalsymb}$ to $\confpair{S_f}{\finalsymb}$ which concludes the proof.
\end{proof}

We define the following two sets:
\begin{itemize}
\item $\covset{\prot} := \max \set{S \subseteq Q \mid \exists \aconfig_0 \in \aconfiginitset, \exists \asymb \in \dataalp, \, \aconfig_0 \step{*} \confpair{S}{\asymb}}$,
\item $\cocovset{\prot, \clause} = \max \set{S \subseteq Q \mid \forall \asymb \in \dataalp, \exists \finalsymb \in \dataok(\clause), {\exists S_f \in \allowedsets(\clause), \, \confpair{S}{\asymb} \step{*} \confpair{S_f}{\finalsymb}}}$.
\end{itemize}
where the $\max$ is for inclusion of sets. Note that for $S \subseteq Q$, it is equivalent that $S$ satisfies the condition $\forall \asymb \in \dataalp, \dots$ in the second set and that there exists a witness execution that starts with a write and therefore is applicable from any $\confpair{S}{\asymb}$ with $\asymb \in \dataalp$.  By convention, we consider that $\max(\emptyset) = \emptyset$, \emph{i.e.}, if $\initialstates = \emptyset$ then $\covset{\prot, \clause} = \emptyset$ and if $Q_-(\prot, \clause) = Q$ then $\cocovset{\prot, \clause} = \emptyset$.

We first prove that both maxima are well-defined because the sets considered are stable under union. 
Let $\aconfig_0, \aconfig_0' \in \aconfiginitset$, $\exec:\aconfig_0 \step{*} \confpair{S}{\asymb}$, $\exec': \aconfig'_0 \step{*} \confpair{S'}{\asymb'}$. 
We show that we can merge $\exec$ and $\exec'$ into a single execution $\confpair{\state{\aconfig_0} \cup \state{\aconfig_0'}}{\datainit} \step{*} \confpair{S \cup S'}{\asymb'}$. By mimicking $\exec$ (without deserting states from $\state{\aconfig_0'}$), we obtain an execution $\confpair{\state{\aconfig_0} \cup \state{\aconfig_0'}}{\datainit} \step{*} \confpair{S \cup \state{\aconfig_0'}}{\asymb}$. By mimicking $\exec'$, we obtain an execution $\confpair{S \cup \state{\aconfig_0'}}{\asymb} \step{*} \confpair{S \cup S'}{\asymb'}$ (because $\prot$ is "uninitialized", $\exec'$ starts with a write). Therefore $S \cup S'$ is in the set $\set{S \subseteq Q \mid \exists \aconfig_0 \in \aconfiginitset, \exists \asymb \in \dataalp, \, \aconfig_0 \step{*} \confpair{S}{\asymb}}$, proving that it is closed under union. 

For $\cocovset{\prot,\clause}$, we suppose that we have $S, S' \subseteq \states$ that satisfy the condition of the set, and we prove that $S \cup S'$ does as well.  
Let $\asymb \in \dataalp$. By hypothesis on $S$ applied with $\asymb$, there exist $S_f$, $\finalsymb$ such that $\confpair{S}{\asymb} \step{*} \confpair{S_f}{\finalsymb}$, and therefore $\confpair{S \cup S'}{\asymb} \step{*} \confpair{S_f \cup S'}{\finalsymb}$. By hypothesis on $S'$ applied with $\finalsymb$, there exist $S_f'$, $\finalsymb'$ such that $\confpair{S'}{\finalsymb} \step{*} \confpair{S_f'}{\finalsymb'}$. Therefore we also have $\confpair{S_f \cup S'}{\finalsymb} \step{*} \confpair{S_f \cup S_f'}{\finalsymb'}$, which combined with the previous execution provides a witness that $S \cup S'$ is in the set.

\begin{algorithm}[p]
\SetKwRepeat{Do}{do}{while}
\SetKwProg{Fn}{Function}{:}{}
\SetKwFunction{solvednfPRP}{DNFPRP{}\_Oneregister\_Uninit}
\SetKwFunction{Coverable}{Compute\_$\covset{\prot}$}
\SetKwFunction{PreviousSymbol}{PreviousSymbol}
\SetKwFunction{coCoverable}{Compute\_$\cocovset{\prot, \clause}$}
\SetKwFor{Lfp}{Until}{reaches a fixpoint do}{endFixpoint}
\Fn{\solvednfPRP{$\prot$}}{
\For{$\clause$ clause of $\psi$}{  
$\prot_{\clause} \assignalgo \prot$ \tcp*[r]{copy of $\prot$ that will be modified}
\Lfp{$\states(\prot_{\clause})$}{ 
  $Q(\prot_{\clause}) \assignalgo Q(\prot_{\clause}) \cap \covset{\prot_{\clause}} \cap \cocovset{\prot_{\clause}, \clause}$ \nllabel{line:new_protocol_less_states_target_onereg} \tcp*[r]{modifies $\prot_{\clause}$}
}
\lIf{$Q_+(\clause) \subseteq \states(\prot_{\clause}) \ne \emptyset$}{Accept}
}
Reject \;
}
\Fn{\Coverable}{
 $S \assignalgo \initialstates$ \;
  \Lfp{$S$}{ 
$S \assignalgo S \cup \set{q' \mid  \exists q \in S, \exists \asymb \in \dataalp, \, (q,\writeact{}{\asymb},q') \in \transitions}$ \nllabel{line:add_write}\;
  $S \assignalgo S \cup \set{q' \mid \exists q, q_1, q_2 \in S, \exists \asymb,  \, (q, \rlreadact{}{\asymb}, q') \in \transitions, (q_1, \writeact{}{\asymb}, q_2) \in \transitions}$ \nllabel{line:add_read}\; }
  \Return $S$ \;
}
\Fn{\coCoverable}{
  \leIf{\PreviousSymbol{$Q \setminus Q_-(\clause), \dataok$} $\ne$ ``\texttt{\upshape Not found}''}{
    $S \assignalgo $\PreviousSymbol{$Q \setminus Q_-(\clause), \dataok$}\;  
  }
  {
    \Return $\emptyset$  
  }
  \Lfp{$S$}{
    \If{\PreviousSymbol{$S, \datawrite$}$ \ne $``\texttt{\upshape Not found}''}{
    $S \assignalgo $\PreviousSymbol{$S, \datawrite$} \;
  }
  }
  \Return $S$\;
}
\Fn{\PreviousSymbol{$S, \mathsf{Symbols}$}}{ 
  $\mathsf{Found} \assignalgo \mathsf{False}$ \;
  \For{$\asymb \in \mathsf{Symbols}$}{ 
    $T \assignalgo S$\;
    \Lfp{$T$}{
        $T \assignalgo T \cup \set{q \in \states \mid \exists q' \in T, (q, \rlreadact{}{\asymb}, q') \in \transitions}$ \; 
      }
    \If{there exist $q \in \states$, $q' \in T$ s.t. $(q, \writeact{}{\asymb}, q') \in \transitions$}{
      $S \assignalgo T \cup \set{q}$ \;
      $\mathsf{Found} \assignalgo \mathsf{True}$ \;
    } 
  }
  \leIf{$\mathsf{Found}$}{\Return $S$}{\Return ``\texttt{Not found}''
}
}

\caption{A polynomial-time algorithm for \dnfPRP with $\rdim = 1$}
\label{algo:ptime_cover_onereg}
\end{algorithm}

Algorithm~\ref{algo:ptime_cover_onereg} provides functions computing $\covset{\prot,\clause}$ and $\cocovset{\clause,\prot}$ along with the function solving \dnfPRP when the protocol is "uninitialized" and $\rdim = 1$.

First, we prove that \texttt{Compute\_$\covset{\prot}$} returns $\covset{\prot}$. By induction, any state added in $S$ in \texttt{Compute\_$\covset{\prot}$} are in $\covset{\prot}$. Indeed, any state that can be covered from a state in $\covset{\prot}$ using a write transition is in $\covset{\prot}$. Similarly, any state that can be covered from a state in $\covset{\prot}$ using a read transition which symbol may be written from $\covset{\prot}$ is in $\covset{\prot}$: first write the corresponding value then read it (all states of $\covset{\prot}$ can be covered in a single, common execution). Conversely, for any execution $\exec: \aconfig_0 \step{*} \aconfig$, every state appearing in $\exec$ is added to $S$. Observe that any execution may be split into phases, where a phase starts with a step writing a symbol then performs some number (possibly zero) of steps reading the symbol. We therefore process by induction on the number of such phases. The initialization comes from $\state{\aconfig_0} \subseteq \initialstates$. Let $\asymb$ the symbol of the last phase in $\exec$, and suppose that all states appearing before this last phase are added to $S$. The write transition of the phase is detected at line~\ref{line:add_write} of the iteration and the corresponding destination state is added to $S$. This write transition is now a witness for $\asymb$ at line~\ref{line:add_read}, allowing every read transition appearing in the phase to be detected in this iteration. 

We now claim that \texttt{Compute\_$\cocovset{\prot,\clause}$} returns $\cocovset{\prot,\clause}$. If $S$ satisfies the condition in $\cocovset{\prot,\clause}$ then one can go from $\confpair{S}{*}$ to a configuration satisfies the clause with an execution starting with a write. Again, this execution may be split into phases, each phase being composed of a write of a symbol followed by some reads of this symbol. The symbol of the last phase must be in $\dataok$ as the last configuration satisfies $\clause$. Therefore, by induction, all states appearing in such an execution appear in some $\PreviousSymbol$ computation in \texttt{Compute\_$\cocovset{\prot,\clause}$}, and the states returned include all states of the execution. Conversely, given a computation of \texttt{Compute\_$\cocovset{\prot,\clause}$} that returns $S$, one may by reversed induction build an execution that covers every state in $S$ and ends on a configuration satisfying $\clause$. All in all, we have proven that \texttt{Compute\_$\cocovset{\prot,\clause}$} computes $\cocovset{\prot, \clause}$.



We will now prove that \texttt{DNFPRP\_Oneregister\_Uninit} of Algorithm~\ref{algo:ptime_cover_onereg} solves \dnfPRP for "uninitialized" protocols with $\rdim = 1$.  First, suppose that the algorithm accepts during the iteration corresponding to clause $\clause$. It ends with a protocol $\prot_{\clause}$ such that $Q_+(\clause) \subseteq \states(\prot_{\clause}) = \covset{\prot_{\clause}} \cap \cocovset{\prot_{\clause}, \clause}$. 
In this protocol, there exist $\aconfig_0 \in \aconfiginitset$ and $\asymb \in \dataalp$ such that $\aconfig_0 \step{*} \confpair{\covset{\prot_{\clause}}}{\asymb}$; since $\covset{\prot_{\clause}} = \cocovset{\prot_{\clause}, \clause}$ we also have $\confpair{\covset{\prot_{\clause}}}{\asymb} \step{*} \confpair{S_f}{\finalsymb} \models \clause$ and the instance is positive. 

Suppose now that the instance is positive. There exist a clause $\clause$ in $\psi$ and a witness execution $\exec: \aconfig_0 \step{*} \confpair{S_f}{\finalsymb}$ with $\finalsymb \in \dataok(\clause)$ and $S_f \in \allowedsets(\clause)$. 
Let $S$ the set of states appearing in $\exec$. By induction, we have that $S \subseteq \covset{\prot_{\clause}} \cap \cocovset{\prot_{\clause}, \clause}$ at every iteration of \texttt{DNFPRP{}\_Oneregister\_Uninit}, because $\exec$ remains a witness of both inclusions at every iteration. Moreover, $Q_+(\clause) \subseteq S$ therefore the algorithm accepts.

\subsection{Proof of Proposition~\ref{prop:target_uninit_nphard_roundless}}
\label{appendix:target_uninit_nphard_roundless}

\targetnpharduninit*


Once again, we provide a reduction from 3-\SAT. Consider a 3-CNF formula $\phi = \bigland_{i=1}^{m} l_{i,1} \lor l_{i,2} \lor l_{i,3}$ over $n$ variables $x_1, \dots, x_n$
where, for all $i \in \nset{m}$, for all $k \in \nset{3}$, $l_{i,k} \in \set{x_j, \neg x_j \mid j \in \nset{n}}$. 

We define an instance of the "uninitialized" \TARGET $(\protsat{\phi}, \targetstate)$ which is positive if and only if $\phi$ is satisfiable.   
Let $\rdim := n$, \emph{i.e.}, the protocol has a register $\rlreg{i}$ for each $i \in \nset{n}$. Each register can have values $\sattrue$ and $\satfalse$ (along with $\datainit$ which cannot be read nor written). 
A depiction of the protocol can be found in Figure~\ref{fig:prot_3SAT_to_uninit_target}. 

  \begin{figure}[tb]
    \centering
  \resizebox{.95\linewidth}{!}{
    \begin{tikzpicture}[node distance = 1.5cm, auto, y = 0.8cm, x = 0.7cm]
\tikzstyle{every node}=[font=\small]
\tikzstyle{every state} = [minimum size = 0.8cm, inner sep = 0.05cm]

\node (q0) [state] at (-0.25,0) {$q_0$}; 
\node(testc1) [state] at (1.5,0) {$C_1?$};
\node (c1ok) [state] at (6,0) {$C_2?$};
\node (susp) at (8.125,0) {{\large \dots}};
\node (testcm) [state] at (10.25,0) {$C_m?$};
\node (qf) [state] at (14.75,0) {$\errorstate$};
\node (int1) [state, rectangle] at (3.75,2) {$\testsat{l_{1,1}}$}; 
\node (int2) [state, rectangle] at (3.75,0) {$\testsat{l_{1,2}}$}; 
\node (int3) [state, rectangle] at (3.75,-2) {$\testsat{l_{1,3}}$}; 
\node (int1bis) [state, rectangle] at (12.5,2) {$\testsat{l_{m,1}}$}; 
\node (int2bis) [state, rectangle] at (12.5,0) {$\testsat{l_{m,2}}$}; 
\node (int3bis) [state, rectangle] at (12.5,-2) {$\testsat{l_{m,3}}$}; 
\node (subprot) [state, rectangle] at (0, -4.5) {$\testsat{x_j}$};
\node (subprot2) [state, rectangle] at (8, -4.5) {$\testsat{\neg x_j}$};
\node (assign) at (1.7,-4.5) {$:=$};
\node (test1) [state, minimum size = 0.6cm] at (3.5,-4.5) {};
\node (test2) [state, minimum size = 0.6cm] at (5.5,-4.5) {};
\node (assign2) at (9.7,-4.5) {$:=$};
\node (test3) [state, minimum size = 0.6cm] at (11.5,-4.5) {};
\node (test4) [state, minimum size = 0.6cm] at (13.5,-4.5) {};

\path[-latex', thick]
    (q0.180) edge[latex'-] +(180:4mm)
    (q0) edge[loop above] node[align = center]{{\small $\forall j \in \nset{n}$} \\ {\small $\writeact{\rlreg{j}}{\sattrue}$}} ()
    (q0) edge[loop below] node[align = center]{{\small $\writeact{\rlreg{j}}{\satfalse}$} \\{\small $\forall j \in \nset{n}$}} ()
    (q0) edge (testc1)
    (testc1) edge (int1)
    (testc1) edge (int2)
    (testc1) edge (int3)
    (testcm) edge (int1bis)
    (testcm) edge (int2bis)
    (testcm) edge (int3bis)
    (int1bis) edge (qf)
    (int2bis) edge (qf)
    (int3bis) edge (qf)
    (int1) edge (c1ok)
    (int2) edge (c1ok)
    (int3) edge (c1ok)
    (test1.180) edge[latex'-] +(180:4mm)
    (test3.180) edge[latex'-] +(180:4mm)
    (test2.0) edge[-latex'] +(0:4mm)   
    (test4.0) edge[-latex'] +(0:4mm)   
    (test1) edge[bend left] node[above] {{\small $\rlreadact{\rlreg{j}}{\sattrue}$}} (test2)
    (test3) edge[bend left] node[above] {{\small $\rlreadact{\rlreg{j}}{\satfalse}$}} (test4);

\path[dashed] 
    (c1ok) edge (6 + 0.5*3, 0.5*3)
    (c1ok) edge (6+0.5*3, 0)
    (c1ok) edge (6 + 0.5*3, -0.5*3);
\path[dashed, -latex']
    (10.25-0.5*3, 0.5*3) edge (testcm)
    (10.25-0.5*3, 0) edge (testcm)
    (10.25-0.5*3, -0.5*3) edge (testcm);
\end{tikzpicture}
  }
 \caption{The protocol $\protsat{\phi}$ for \NP-hardness of "uninitialized" \TARGET.}
\label{fig:prot_3SAT_to_uninit_target}
  \end{figure}

Suppose first that $\phi$ is satisfiable by an assignment $\nu$. For all $i \in \nset{m}$, there exists $k(i) \in \nset{3}$ such that $\nu(l_{i,k(i)}) = \true$. Consider the execution that writes symbols according to $\nu$, then deserts $q_0$ to go to $C_1?$, and one by one deserts all $C_i?$-s through states $\testsat{l_{i,k(i)}}$. This execution goes from $\confpair{q_0}{\datainit^{\regset{}}}$ to $\confpair{\targetstate}{\dvec_f}$ hence the instance of \TARGET is positive. 
Conversely, suppose that there exists such an execution $\exec: \aconfig_0 \step{*} \confpair{\targetstate}{\dvec}$. Let $\nu$ be the valuation corresponding to the register values when $\exec$ deserts $q_0$ for the last time. From this point onwards, $\exec$ successively deserts all $C_i?$, hence for all $i \in \nset{m}$, there exists $k(i) \in \nset{3}$ such that $\nu(l_{i,k(i)}) = \true$, proving that $\phi$ is satisfied by $\nu$. 

\section{Round-based Register Protocols}

\subsection{Formal Definition of Atomic Presence Constraints}
\label{def:presence_constraints_roundbased}
We first define some more precise notions to refer to parts of "presence constraints". 
A ""term"" is of the form $m$ or $k{+}m$ with $m \in \NN$ and $k$ a free variable. An ""atomic proposition"" is either of the form ``$(q,t) \populated$'' with $t$ a "term" and $q \in \states$ or of the form ``$\reg{t}{\regid} \contains \asymb$'' with $t$ a "term", $\regid \in \regint$ and $\asymb \in \dataalp$. 
A ""literal"" is either an atomic proposition or the negation of an atomic proposition. A ""proposition"" is a Boolean combination of "atomic propositions" that has at most one free variable.
An ""atomic presence constraint"" is either a ""closed"" "proposition" (no free variables), or of the form ``$\exists k \, \phi$'' or ``$\forall k \, \phi$'' where $\phi$ is a "proposition" with $k$ as a free variable. 
 A ""presence constraint@@roundbased"" is a Boolean combination of "atomic presence constraints".

\subsection{Proof of Lemma~\ref{lem:combiningfootprints}}
\label{appendix:combiningfootprints}

We prove the following, more general statement. 
\begin{lemma}
\label{lem:combiningfootprintsextended}
Let  $K \in \NN$, $w \geq \vrange{-}1$, $(\afootprint_k)_{k \leq K}$ and $(T_k)_{k \leq K{-}1}$ such that:
\begin{itemize}
\item for all $k \leq K$, $\afootprint_k$ is a "footprint" on $\iset{k{-}w}{k}$,
\item for all $k \leq K{-}1$, $T_k$ is a "footprint" on $\iset{k{-}w}{k{+}1}$,
\item for all $k \leq K{-}1$, $\fpproj{k{-}w}{k}{T_k} = \afootprint_k$,
\item for all $k \leq K{-}1$, $\fpproj{k{-}w{+}1}{k{+}1}{T_k} = \afootprint_k$.
\end{itemize}
 There exists an execution $\exec$ such that, for all $k \leq K$, $\fpproj{k{-}w}{k}{\exec} = \afootprint_k$. 
\end{lemma}

We start by proving the following lemma:
\begin{lemma}
\label{lem:upwardsuccessormiddleproj}
Let $w \geq \vrange$, $k \in \NN$, $\afootprint_-$ a "footprint" on $\iset{k{-}w}{k}$ and $ \afootprint_+$ a "footprint" on $\iset{k{-}w{+}1}{k{+}1}$ such that $\fpproj{k{-}w{+}1}{k}{\afootprint_-} = \fpproj{k{-}w{+}1}{k}{\afootprint_+}$. There exists $T$ a footprint on $\iset{k{-}w}{k{+}1}$ such that $\fpproj{k{-}w}{k}{T} = \afootprint_-$ and $\fpproj{k{-}w{+}1}{k{+}1}{T} = \afootprint_+$.
\end{lemma}
\begin{proof}
Let $\afootprint_{com} := \fpproj{k{-}w{+}1}{k}{\afootprint_-} = \fpproj{k{-}w{+}1}{k}{\afootprint_+}$. We proceed by induction on the number of steps in $\afootprint_{com}$. 

First if $\afootprint_{com}$ is the dummy footprint with no steps, then all steps in $\afootprint_-$ are at round $k{-}w$ and steps in $\afootprint_+$ are at round $k{+}1$. It suffices to consider $T$ that first copies the behavior of $\afootprint_-$ and then the behavior of $\afootprint_+$: steps at round $k{-}w$ cannot depend on the information of rounds $> k {-} w$, and steps at round $k{+}1$ cannot depend on the information of rounds $< k{-}w$ because $w \geq \vrange$.

Assume that the property is true if $\afootprint_{com}$ has $m$ steps, and suppose that $\afootprint_{com}$ has $m{+}1$ steps. We decompose $\afootprint_- = t_-, \move, s_-$ and $\afootprint_+ = t_+, \move, s_+$ where $t_-$ and $t_+$ coincide on rounds $k{-}w{+}1$ to $k$ and their projection on these rounds has exactly $m$ steps, $\move$ is the move of the $m{+}1$-th step of $\afootprint_{com}$, and $s_-$ and $s_+$ have no step at rounds $k{-}w{+}1$ to $k$. By induction hypothesis, there exists $t$ such that $\fpproj{k{-}w}{k}{t} = t_-$ and $\fpproj{k{-}w{+}1}{k{+}1}{t} = t_+$. By applying the property for $m=0$, there exists $s$ such that $\fpproj{k{-}w}{k}{s} = s_-$ and $\fpproj{k{-}w{+}1}{k{+}1}{s} = s_+$. Letting $T := t,\move,s$ (with $\move$ deserting if and only if it was "deserting" in $\afootprint_{com}$) concludes the proof.
\end{proof}
 
We now prove Lemma~\ref{lem:combiningfootprintsextended}.
We proceed by induction on $K$.
First, if $K=0$, "footprint" $\afootprint_0$ only has moves at round $0$ and may be seen as an execution. 
Suppose that the property is true for $K$, and consider $(\afootprint_k)_{k \leq K{+}1}$, $(T_k)_{k \leq K}$ satisfying the hypothesis. For all $k \leq K{-}1$, $T_k$ and $T_{k{+}1}$ both have projection $\tau_{k{+}1}$ on rounds $\iset{k{-}w{+}1}{k{+}1}$, hence thanks to Lemma~\ref{lem:upwardsuccessormiddleproj} applied with $w' := w{+}1$ and $k' := k{+}1$, there exists $U_k$ on rounds $\iset{k{-}w}{k{+}2}$ that projects to $T_k$ and $T_{k{+}1}$ on $\iset{k{-}w}{k{+}1}$ and $\iset{k{-}w{+}1}{k{+}2}$ respectively. 
 By applying the induction hypothesis on $(T_k)$ and $(U_k)$ with $K' := K{-}1$, there exists an execution $\exec$ such that, for all $k \leq K$, $\fpproj{k{-}w}{k{+}1}{\exec} = T_k$; this implies that, for all $k \leq K{+}1$, $\fpproj{k{-}w}{k}{\exec} = \afootprint_k$, concluding the proof of Lemma~\ref{lem:combiningfootprintsextended}. Applying Lemma~\ref{lem:combiningfootprintsextended} with $w := \vrange{-}1$ gives Lemma~\ref{lem:combiningfootprints}.

\subsection{Technical Details about Algorithm~\ref{algo:pspace}}
\label{appendix:details_algo_pspace}

Here, we describe in full details how Algorithm~\ref{algo:pspace} handles the "presence constraint". The pseudocode of the three functions used in Algorithm~\ref{algo:pspace} can be found in Algorithm~\ref{algo:functions}.

For $\psi$ a "presence constraint", we write \AP $\intro*\Atomcons{\psi}$ for the set containing all "atomic presence constraints" in $\psi$ as well as their negations.  
For $\phi$ a closed "proposition", we write \AP $\intro*\Atomprop{\phi}$ for the set of "atomic propositions" in $\phi$.
Given a set $S$ of "propositions" or "presence constraints", we write \AP $\intro* \negset{S}:= S \cup \{ \neg P \mid P \in S\}$ for the set containing all elements in $S$ and the negations of all elements in $S$. 

\begin{algorithm}[htbp]
\SetKwProg{Fn}{Function}{:}{}
\SetKwFunction{InitFun}{\initalgo}
\SetKwFunction{OnestepFun}{\onestepalgo}
\SetKwFunction{TestFun}{\testalgo}
\Fn{\InitFun{$E,U,C$} \nllabel{line_def_initfun}}{ 
\tcc{Sets containing what needs to be checked: $U$ and $E$ contain respectively universally and existentially quantified "atomic presence constraints", $C$ contains closed "literals"}
Guess $X \subseteq \negset{\Atomcons{\psi}}$ s.t. $\psi$ is true when all APCs in $X$ are true  \nllabel{line_guess_apcs} \;

\For{$P$ in $X$}{
  \For{$\phi$ closed "atomic proposition" in $P$ }{
  \tcc{$\phi$ refers to constant rounds only}
  \lIf{$\phi$ guessed to be true \nllabel{line:guess_closed_aps}}{
    Add $\phi$ to $C$ \nllabel{line:add_to_C_1}; 
    Replace $\phi$ by true in $P$ \nllabel{line:simplify_1}
  }
  \lElse {
  Add $\neg \phi$ to $C$ \nllabel{line:add_to_C_2};
  Replace $\phi$ by false in $P$ \nllabel{line:simplify_2}
  }
  }
\If{$P$ is a closed "proposition"}{
  Check that $P$ is true with guessed values of "atomic propositions"  \nllabel{line:check_closed}\;
}
  \lIf{$P$ universal}{Add $P$ to $U$ \nllabel{line_init_u}}
  \lIf{$P$ existential}{Add $P$ to $E$ \nllabel{line_init_e}} 
} 
}
\Fn{\OnestepFun{$E,U,C, \localconfig$} \nllabel{line_def_onestepfun}}{
\For{``$\forall l \, \phi$'' in $U$ \nllabel{line_guess_forall_1}}{
  Guess  $\mathcal{L} \subseteq \negset{\Atomprop{\phi[l \shortleftarrow k]}}$ s.t. $\phi[l \shortleftarrow k]$ is true when all literals in $\mathcal{L}$ are true \nllabel{line_guess_forall_2}\; 
  Add all "literals" in $\mathcal{L}$ to $C$ \nllabel{line_guess_forall_3}\nllabel{line:add_to_C_3}\;
}
\For{``$\exists l \, \phi$'' in $E$ \nllabel{line_guess_exists_1}}{
  \If{$\phi[l \shortleftarrow k]$ guessed to be true \nllabel{line_guess_exists_3}}{
     Guess  $\mathcal{L} \subseteq \negset{\Atomprop{\phi[l \shortrightarrow k]}}$ s.t. $\phi[l \shortleftarrow k]$ is true when all literals in $\mathcal{L}$ are true \nllabel{line_guess_exists_4} \; 
    Add all "literals" in $\mathcal{L}$ to $C$ \nllabel{line_guess_exists_5} \nllabel{line:add_to_C_4};
    Remove ``$\exists l \, \phi$'' from $E$ \nllabel{line_guess_exists_6}\;
  }
}
\For{$\phi$ in $C$ about round $k$ \nllabel{line_check_c_1}}{
    \tcp{$\phi$ is of the form (negation of) ``$(q,k) \populated$'', or (negation of) ``$\reg{k}{\regid}  \contains \asymb$''}
    Check that $\phi$ is satisfied in $\localconfig$ \nllabel{line_check_c_3};
    Remove $\phi$ from $C$ \nllabel{line_check_c_4}\;
}
}
\Fn{\TestFun{$E,U,C, \localconfig$} \nllabel{line_def_testfun}}{
\lIf{$E \ne \emptyset$ \nllabel{line_test_emptiness}}{\Return false \nllabel{line_check_empty}}
\For{$\phi \in C$ or ``$\forall l \, \phi$'' in $U$ \nllabel{line_check_infinite_1}}{
    \If{$\confpair{\emptyset}{\datainit^{\regset{}}} \not \models \phi$}{\Return false \nllabel{line_check_infinite_2}\tcp*[r]{Execution cannot stop at round $k$}}
}
\Return true \;
}
\caption{The functions at \nlref{line_initfun}, \nlref{line_onestepfun} and \nlref{line_testfun} of Algorithm~\ref{algo:pspace}}\label{algo:functions}
\end{algorithm}

\subsubsection{Function \texttt{\initalgo} (\nlref{line_def_initfun}):}
At \nlref{line_guess_apcs}, we guess a partial assignment over "atomic presence constraints" that makes $\psi$ true. Recall that "atomic presence constraints" either are closed "propositions" or of the form ``$\forall l \, \phi$'' or ``$\exists l \, \phi$'' with $\phi$ a "proposition" that has $l$ as free variable. We see this assignment as a set of "atomic presence constraints" which, when set to true, make $\psi$ true. Note that negations of "atomic presence constraints" are "atomic presence constraints". All closed "atomic propositions" refer to constant rounds; guess which ones are true (\nlref{line:guess_closed_aps}). This simplifies all closed "propositions" in $X$ to either true of false: if any of them is false, we reject (\nlref{line:check_closed}). 
We put universally quantified element of $X$ in $U$ 
(\nlref{line_init_u}) and existentially quantified ones in $E$
(\nlref{line_init_e}).

\subsubsection{Function \texttt{\onestepalgo} (\nlref{line_def_onestepfun}):}
The universal "atomic presence constraints" are checked at every round (\mulnlref{line_guess_forall_1}{line_guess_forall_3}), 
while for each existential "atomic presence constraints" is checked at a round chosen non-deterministically  (\mulnlref{line_guess_exists_1}{line_guess_exists_6}). 
When checking a proposition, we guess which "literals" make them true, and put these "literals" in $C$ to be checked later.
Moreover, we check at round $k$ all "literals" in $C$ that are about round $k$ (at \twonlref{line_check_c_1}{line_check_c_4}). Note that all "literals" in $C$ are closed formulas hence their "terms" are constant integers. 

\subsubsection{Function \texttt{\testalgo} (\nlref{line_def_testfun}):} In this functon, we check whether we can stop the execution at round $k$, leaving all rounds $\geq k{+}1$ untouched. First, we check that $E$ is empty. This means that a round has been guessed for every existential formula that has been in $E$. Moreover, we check that remaining formulas in $C$ and $U$ would be satisfied at rounds $\geq k{+}1$ if these rounds are left untouched by the execution, which is done in \mulnlref{line_check_infinite_1}{line_check_infinite_2}. 
The test is expressed under the condition $\confpair{\emptyset}{\datainit^{\regset{}}} \models \phi$ (although $\confpair{\emptyset}{\datainit^{\regset{}}}$ is technically not a configuration as it has zero processes), and is implemented as follows. Any formula $\phi$ that is in $C$ at the end of iteration $k$ is about round $l \geq k{+}1$ at this stage, and we check that $\phi$ is either of the form ``$\neg((q,l) \populated)$'' or of the form ``$\reg{l}{\regid} \contains \datainit$''. 
A universal "presence constraint" $\forall l \, \phi$ must be satisfied on arbitrarily large rounds $\geq k{+}1$, and we check that we obtain true by setting in $\phi$ all ``$(q,t) \populated$'' to false, ``$\reg{t}{\regid}\contains \datainit$'' to true and ``$\reg{t}{\regid}\contains \asymb$'' to false for $\asymb \ne \datainit$. 

\begin{example}
Consider $\phi_1 := \forall l \, ((q,l) \populated) \lor (\reg{l}{\regid} \contains \datainit)$ and $\phi_2 := \forall l \, (\reg{l}{\regid} \contains \asymb)$ with $\asymb \ne \datainit$. One has $\confpair{\emptyset}{\datainit^{\regset{}}} \models phi_1$, but $\confpair{\emptyset}{\datainit^{\regset{}}} \not \models \phi_2$. There is no hope of finding a $\aconfig \in \areach{\aconfiginitset}$ such that $\aconfig \models \phi_2$. 
\end{example}

\subsection{Proof of Correctness of the Algorithm}
\label{appendix:proof_correctness}
\correctnessalgo*

First, consider a computation of the algorithm that accepts at round $K \in \NN$.
For all $k \in \iset{0}{K}$, let $\tau_k$ denote the "footprint" on $\iset{k{-}\vrange{+}1}{k}$ guessed by the algorithm during iteration $k$. 
By applying Lemma~\ref{lem:combiningfootprints}, there exist $\aconfig_0 \in \aconfiginitset$ and an execution $\exec: \aconfig_0 \step{*} \aconfig$ such that, for all $k \leq K$, $\fpproj{k{-}\vrange{+}1}{k}{\exec} = \tau_k$. Moreover, $\exec$ leaves rounds $\geq K$ untouched.

\begin{lemma}
\label{fact:all_presence_constraints_checked}
For every formula $P$ that was in $U$, $E$ or $C$ at any point throughout the computation, one has $\aconfig \models P$. 
\end{lemma}
\begin{proof}
Let $L$ be a literal that has been in $C$ at some point.  If it was removed from $C$ at \nlref{line_check_c_4}, then $C$ is satisfied by $\localconfig$ hence by $\aconfig$. If it has remained in $C$ until the end, then it is about round $l \geq K{+}1$ and $\confpair{\emptyset}{\datainit^{\regset{}}} \models L$, hence $\aconfig \models L$. 

Consider ``$\exists l \, \phi$'' that has appeared in $E$ at some point; it was added to $E$ at  \nlref{line_init_e}. 
At some iteration $k$, ``$\exists l \, \phi$'' is removed from $E$ at \nlref{line_guess_exists_6}. All literals guessed at \nlref{line_guess_exists_4} are added to $C$ at \nlref{line_guess_exists_5} hence are satisfied by $\aconfig$, thus $\aconfig \models \phi[l \shortleftarrow k]$ and $\aconfig \models \exists l \, \phi$.

Similarly, consider ``$\forall l \, \phi$'' that has appeared in $U$ at some point. By the same argument, for all $k \leq K$, $\aconfig \models \phi[l \shortleftarrow k]$. Also, thanks to the verification at \mulnlref{line_check_infinite_1}{line_check_infinite_2}, for all $k \geq K{+}1$, $\aconfig \models \phi[l \shortleftarrow k]$, which proves that $\aconfig \models \exists l \, \phi$. 
\end{proof}


The previous lemma proves that all APCs guessed at \nlref{line_guess_apcs} are satisfied by $\aconfig$. Note that the simplification at \twonlref{line:simplify_1}{line:simplify_2} does not change the truth value of APC $P$. Finally, we have $\aconfig \models \psi$. 

We now prove the converse implication: suppose that there exists $\exec: \aconfig_0 \step{*} \aconfig$ with $\aconfig \models \psi$.  
Since $\exec$ is a finite "execution", there exists $K$ such that $\aconfig$ has no "move" with "effect" on rounds $>K$. 
We build an accepting computation of the algorithm as follows.
First, the computation of the algorithm guesses $\aconfig_0$ as initial configuration.
At \nlref{line_guess_apcs}, it guesses APCs $P$ such that $\aconfig \models P$.
At \nlref{line:guess_closed_aps}, it guesses the truth value of closed APs in $\aconfig$, so that all formulas added to $C$, $E$ and $U$ are satisfied by $\aconfig$. 
In the loop on $k$, it guesses $\exec$ "footprint" by "footprint". At execution $k$, the "local configuration" $\localconfig$ obtained is equal to $\fpproj{k{-}\vrange}{k}{\aconfig}$.
Formulas in $E$ and $U$ do not have closed terms, and since quantified terms are of the form $l{+}m$ with $l$ a free variable, "literals" added to $C$ at iteration $k$ refer to rounds $\geq k$; thanks to \mulnlref{line_check_c_1}{line_check_c_4}, at the end of iteration $k$, all "literals" in $C$ are about rounds $\geq k{+}1$. At the end of iteration $K$ (or an earlier iteration), all formulas in $C$ and $U$ are satisfied by $\aconfig$ (which is blank after round $K$) hence \texttt{\testalgo}$(E,U,C,\localconfig)$
succeeds and the computation accepts.  This concluded the proof of Proposition~\ref{prop:correctness_algo_pspace}.

\subsection{Proof of Proposition~\ref{prop:polynomial_space}}
\label{appendix:polynomial_space}

\algopolyspace*

We first prove that footprints may be stored in polynomial space.

\begin{restatable}{lemma}{normalformroundbased}
\label{lem:normal_form_roundbased}
For all $\aconfig \in \aconfigs, \aconfig' \in \areach{\aconfig}$, there exists $\exec: \aconfig \step{*} \aconfig'$ s.t., for all $k$,  $\fpproj{k{-}\vrange}{k}{\exec}$ is storable in space $O((|Q|^2 \, (\vrange{+}1)^2 \,  \log(\rdim \, |\dataalp|)))$.
\end{restatable}

More specifically, we will prove that $\fpproj{k{-}\vrange}{k}{\exec}$ is storable in space $O((|Q|^2 \, (\vrange{+}1)^2 + |Q| \, (\vrange{+}1)^2 \,  \log(\rdim \, |\dataalp|)))$.
Similarly to the roundless case, we introduce a notion of normal form.
An execution $\exec$ is in ""normal form"" if for every step in $\exec$, one the following conditions is satisfied:
\begin{itemize}
\item the step writes a symbol to a register, and this symbol is later read by another step, or
\item it "deserts" the source location, or
\item its destination location was not "populated" before the step and has never been populated before in the execution. 
\end{itemize}
Note that the last two conditions combined imply that a given "location" is deserted at most once, as it cannot be deserted and then populated again. 

\begin{lemma}
\label{lem:normal_form_exists_roundbased}
For all execution $\exec: \aconfig \step{*} \aconfig'$, there exists a execution $\ti{\exec}: \aconfig \step{*} \aconfig'$ that is in "normal form".
\end{lemma}
\begin{proof}
It suffices to iteratively:
\begin{itemize}
\item remove any read or increment that is "non-deserting" and does not cover a new location,
\item remove any write that is "non-deserting", does not "populate" a new location and whose written symbol is never read,
\item turn into "non-deserting" any "deserting" step that "deserts" a location which is later "populated" again.
\end{itemize}\end{proof}

\begin{lemma}
\label{lem:normal_form_limited_moves_roundbased}
An execution in "normal form" has at most $|Q| (2 \vrange{+}5)$ steps on a given round $k$. 
\end{lemma}
\begin{proof}
First, any read or increment step at round $k$ either deserts its source location which is never populated again, or populates its destination (\emph{i.e.}, its destination was not populated before the step). However, each location has at most one step populating the location and one deserting the location. Since steps at round $k$ may only desert locations of round $k$ and populate locations at rounds $k$ and $k{+}1$, at most $3 |Q|$ steps at round $k$ either desert or populate a location, among which at most $2 |Q|$ read steps as they may only desert and populate locations of round $k$.
Moreover, any write step at round $k$ that does not populate or desert must be read later, and that has to be by a read step on a round between $k$ and $k{+} \vrange$. Since there are at most $2 |Q| (\vrange{+}1)$ read steps on these rounds, there are at most $2 |Q| (\vrange{+}1)$ writes at round $k$ that do not populate nor desert, hence in total at most $|Q|(2 \vrange{+}5)$ steps at round $k$.
 \end{proof}

\begin{lemma}
\label{lem:footprint_storable_polyspace_roundbased}
If $\exec$ is in "normal form" and $k \in \NN$, then $\fpproj{k{-}\vrange}{k}{\exec}$ is storable in polynomial space $O((|Q|^2 \, (\vrange{+}1)^2 + |Q| \, (\vrange{+}1)^2 \,  \log(\rdim \, |\dataalp|)))$. 
\end{lemma}
\begin{proof}
This "footprint" only has steps at rounds $k{-}\vrange$ to $k$, hence in total at most $(\vrange{+}1) |Q| (2 \vrange {+} 5)$ steps. Since a move can be stored in $O(\log(|Q|) + \log(\dataalp) + \log(\vrange) + \log(\rdim))$ and a local configuration in $O((|Q| + \rdim \, \log(|\dataalp|)) (\vrange{+}1))$ (storing the relative round instead of the absolute one),  a "footprint" on $\iset{k{-}\vrange}{k}$ can be stored in polynomial space $O((|Q|^2 \, (\vrange{+}1)^2 + |Q| \, (\vrange{+}1)^2 \, \log(\rdim \, |\dataalp|)))$.
\end{proof}

Combining Lemmas~\ref{lem:normal_form_exists_roundbased}, \ref{lem:normal_form_limited_moves_roundbased} and \ref{lem:footprint_storable_polyspace_roundbased} proves Lemma~\ref{lem:normal_form_roundbased}. 
Observe that Lemma~\ref{lem:normal_form_roundbased} is only true under the assumption that we do not store the rounds of a footprint in absolute value but in relative value with respect to $k$; otherwise the space used would depend on $k$.

We now prove Proposition~\ref{prop:polynomial_space}.

Thanks to Proposition~\ref{lem:normal_form_roundbased}, one may store $\tau$ and $T$ in polynomial space.
$U$ and $E$ are storable in $O(|\psi|)$, as for every "atomic presence constraint" $\phi$ in $U$ and $E$, either $\phi$ is present in $\psi$ or its negation $\neg \phi$ is. 
Let $M$ be the value of the greatest integer constant in $\psi$, which is in $O(|\psi|)$ thanks to unary encoding of the "terms". A "literal" can get to $C$ in two different ways: during the initialization (\twonlref{line:add_to_C_1}{line:add_to_C_2}) or while processing a "presence constraint" from $U$ or $E$ (\twonlref{line:add_to_C_3}{line:add_to_C_4}). There are at most $O(|\psi|)$ "literals" added to $C$ in the 
initialization. Consider $L$ a "literal" that is added to  $C$ at \nlref{line:add_to_C_3} or \nlref{line:add_to_C_4} during the computation of round $k$. Let $r$ be the round appearing in $L$. Either $r$ is a constant from $\psi$, or it was added at iteration $k' \leq k$, hence $r$ is of the form $k' + m$ with $m \leq M$. In that case, note that $r \geq k$ because otherwise the "literal" would have been removed from $C$ at iteration $r$. Either way, one has $0 \leq r -k \leq M$, hence a given element in $C$ is storable in $O(|\psi|)$. 
Also, elements in $C$ at round $k$ were added to $C$ either at the initialization or at a round in $\iset{k{-}M}{k}$, which bounds the total number of elements in $C$ by $O(M |\psi|) = O(|\psi|^2)$ at any point in the computation, and $C$ is storable in $O(|\psi|^3)$.

\end{document}